\documentclass[
preprint,
 amsmath,amssymb,
 aps, physrev,
 floatfix,
]{revtex4-2}

\usepackage{amsmath,amssymb,amsfonts}
\usepackage{physics}
\usepackage{amsthm}
\usepackage{booktabs}
\usepackage{float}
\usepackage{algorithm}
\restylefloat{algorithm}
\usepackage{algpseudocode}
\newtheorem{theorem}{Theorem}
\newtheorem{lemma}{Lemma}
\newtheorem{proposition}{Proposition}
\newtheorem{corollary}{Corollary}
\newtheorem{remark}{Remark}

\usepackage{graphicx}%
\usepackage{dcolumn}%
\usepackage{bm}%
\usepackage{tikz}
\usepackage{quantikz}
\usepackage{tabularx}
\usepackage{hyperref}

\newcommand{\myheading}[1]{\vspace{1ex}\noindent\textbf{#1}\quad}

\begin{document}

\preprint{APS/arXiv draft}

\title{\textbf{Asymmetric Linear-Combination-of-Unitaries Realization of Quantum Convolution via Modular Adders}
}%

\author{Chen Yang}
\affiliation{%
 Department of Chemistry, Faculty of Science, Kyushu University, 744 Motooka Nishi-ku, Fukuoka 819-0395, Japan
}%
\author{Kodai Kanemaru}%
\affiliation{%
  Fukui Institute for Fundamental Chemistry, Kyoto University, Takano-Nishibiraki-cho 34-4, Sakyou-ku, Kyoto 606-8103, Japan \\
  Department of Complex Systems Science, Graduate School of Informatics, \\
 Nagoya University, Furo-cho, Chikusa-ku, Nagoya, Aichi 464-8601, Japan
}%

\author{Norio Yoshida}
\affiliation{%
 Department of Complex Systems Science, Graduate School of Informatics, \\
 Nagoya University, Furo-cho, Chikusa-ku, Nagoya, Aichi 464-8601, Japan
}%

\author{Sergey Gusarov}
\affiliation{
National Research Council Canada,
Ottawa, ON, Canada
}%

\author{Hiroshi C. Watanabe}
 \email{Contact author: hcwatanabe@chem.kyushu-univ.jp}
\affiliation{
 Department of Chemistry, Faculty of Science, Kyushu University, 744 Motooka Nishi-ku, Fukuoka 819-0395, Japan\\
 Quantum Computing Center, Keio University, 3-14-1 Hiyoshi, Kohoku-ku, Yokohama, Kanagawa 223-8522, Japan\\
 Quantum and Spacetime Research Institute, Kyushu University, 744 Motooka Nishi-ku, Fukuoka 819-0395, Japan
}%

\date{\today}

\begin{abstract}
Discrete circular convolution over $\mathbb{Z}/N\mathbb{Z}$ is a linear operator and can be implemented on quantum hardware within the linear-combination-of-unitaries (LCU) framework. 
In this work, we make this connection explicit through an asymmetric-LCU formulation: circular convolution is the postselected block of a circuit whose controlled-shift unitary is modular addition on computational-basis states. 
The asymmetry is essential: fixing the postselection state to the uniform state $\ket{u}$ while supplying the kernel state $\ket{\bm b}$ as the input ancilla naturally preserves the complex coefficients $b_i$ within the block, whereas a symmetric overlap would yield $|b_i|^2$ weights and erase their phases. 
Accordingly, when $\ket{\bm a}$ and $\ket{\bm b}$ are already supplied by upstream quantum routines, the convolution subroutine requires only the fixed uncompute $\mathrm{PREP}_u^\dagger$, completely avoiding the need for a kernel-dependent inverse preparation $\mathrm{PREP}_b^\dagger$.

We then introduce a reversal matrix $J_n=X^{\otimes n}$ and define reflected shifts $\widetilde{L}_{i,n}=L_{i,n}J_n$. 
This symmetrization yields a recursive operator algebra for convolution that is natively compatible with standard LCU/block-encoding workflows. The resulting symmetrized operator differs from standard circular convolution only by one known input-side $J_n$ layer.
Crucially, for real-valued kernels, the resulting operator $H_n(\bm{b})=\sum_i b_i\widetilde{L}_{i,n}$ is Hermitian, providing a direct Hermitian interface for quantum singular value transformation (QSVT) and related spectral transformations.

Based on this framework, we present a structurally transparent recursive construction, paired with an exactly equivalent optimized bitwise compilation of the same $\mathrm{SELECT}$ block.
Finally, we evaluate implementation trade-offs and resource scaling under explicit cost-model conventions.
\end{abstract}

\maketitle

\section{Introduction}

Whether ``quantum convolution'' is physically meaningful depends on \emph{which map} is being implemented. 
Lomont's no-go argument rules out direct bilinear amplitude multiplication as a physical quantum operation~\cite{Lomont2003}. 
However, discrete convolution, when formulated as a linear operator acting on basis-indexed vectors, is fully compatible with laws of quantum mechanics and can be realized by linear-combination-of-unitaries (LCU) and block-encoding techniques~\cite{Childs2012,Berry2015,Gilyen2019,Low2019}.

Two key lines of prior work established this operator-centric viewpoint. 
First, Zhou and Wang showed that circulant matrices decompose into cyclic shifts and can be implemented via modular addition circuits within an LCU pipeline~\cite{Wang2017}. 
Second, Castelazo \textit{et al.} formulated finite-group convolution as a linear combination of left-regular shifts and provided both LCU-based and Fourier-domain constructions~\cite{Castelazo2022}. 
For the cyclic group $\mathbb Z/N\mathbb Z$, these two perspectives describe the same mathematical object viewed from arithmetic and representation-theoretic angles.
The connection is implicit across the two works but not stated directly at the circuit level: Zhou and Wang emphasize circulant-matrix multiplication, whereas Castelazo \textit{et al.} emphasize the general group-convolution framework~\cite{Wang2017,Castelazo2022}.

A central technical observation of this work is that, for discrete circular convolution over $\mathbb Z/N\mathbb Z$, this equivalence can be written explicitly at the circuit level as an asymmetric LCU/block-encoding statement, in the sense used in quantum signal processing (QSP), quantum singular value transformation (QSVT), and qubitization/LCU frameworks~\cite{Low2017,Gilyen2019,Low2019}. The asymmetry is crucial. The postselection state is the fixed uniform state $\ket{u}$, while the input ancilla is the problem-dependent kernel state $\ket{\bm b}$; this preserves the complex coefficients $b_i$ themselves in the postselected block. By contrast, a same-state symmetric overlap using $\mathrm{PREP}_b$ on both sides would produce $|b_i|^2$ weights, so the phases of the complex amplitudes would cancel. At the formal theorem level we still write the construction in the standard oracle language with $\mathrm{PREP}_b$, but in the supplied-state specialization the convolution subroutine requires only the fixed uncompute $\mathrm{PREP}_u^\dagger$ and no kernel-dependent inverse preparation $\mathrm{PREP}_b^\dagger$. To our knowledge, prior works did not state this bridge in the circular-convolution setting in this exact asymmetric supplied-state form: Zhou and Wang emphasize circulant implementation via modular-addition-based $\mathrm{select}(V)$ circuits, whereas Castelazo \textit{et al.} formulate general group convolution through entry or Fourier-entry oracle access.

After establishing this equivalence, we introduce a reversal matrix $J_n$ and define symmetrized shifts $\widetilde{L}_{i,n}=L_{i,n}J_n$. 
This step exposes a recursive operator structure for the convolution family. 
For real-valued kernels, the symmetrized operator becomes Hermitian, giving a direct Hermitian interface for QSVT and related downstream spectral transformations.

We then present a structural recursive realization of $\mathrm{SELECT}_{\widetilde{L}}$, relate it to an exactly equivalent supplementary bitwise compilation, and use standard quantum Fourier transform (QFT) and ripple-carry adders only as canonical implementation instantiations of the modular-addition shift block.

For clarity, the manuscript makes three claims of novelty, all relative to the circulant-adder viewpoint of Ref.~\cite{Wang2017} and the group-convolution viewpoint of Ref.~\cite{Castelazo2022}:
\begin{enumerate}
\item \textbf{Explicit asymmetric-LCU formulation and supplied-state specialization.}
We make explicit that circular convolution can be realized as an asymmetric postselected LCU block with fixed postselection state $\ket{u}$ and kernel state $\ket{\bm b}$, thereby preserving the complex coefficients $b_i$ themselves rather than the weights $|b_i|^2$.
\item \textbf{$J_n$-symmetrized Hermitian operator pipeline.}
We introduce the symmetrized family $\widetilde{L}_{i,n}=L_{i,n}J_n$ and the corresponding operator $H_n(\bm b)$, which differs from standard convolution only by one known input-side $J_n$ layer and becomes Hermitian for real-valued kernels.
\item \textbf{Recursion-first structural normal form with exact compiled realization.}
We present a structural recursion for the $J_n$-symmetrized reflected-shift family and relate it to an exactly equivalent compiled carry-propagation implementation of the same $\mathrm{SELECT}_{\widetilde L}$ block.
\end{enumerate}
We do not claim novelty for cyclic-shift decompositions themselves, for the identity $\mathrm{ADD}_N=\mathrm{SELECT}_L$, or for individual QFT and ripple-carry adder primitives. We also do not claim an asymptotic improvement in arithmetic complexity over the best standalone modular adder constructions. Rather, the contribution is a sharper operator-level formulation, a structurally transparent recursive normal form, and a direct Hermitian route for inverse-type spectral processing. The QFT and ripple-carry adder sections are therefore included as canonical realizations and comparison points for this common block, rather than as independent algorithmic contributions of the present work.

The remainder of the paper is organized as follows. Section~\ref{sec:preliminaries} reviews the classical and group-theoretic formulations of circular convolution and states the asymmetric-LCU realization via modular addition. Section~\ref{sec:quantum_representation_of_convolution} introduces the $J_n$-symmetrized operator family and its structural recursion. Section~\ref{sec:circuits} presents the direct recursive circuit construction together with standard QFT adder and ripple-carry adder realizations. Section~\ref{sec:complexity} discusses resource scaling under the different input and cost models used in this work. Section~\ref{sec:blockencoding_qsvt} then returns to the Hermitian structure of $H_n(\bm b)$ and its use in block-encoding and deconvolution via QSVT.

\section{Preliminaries and Asymmetric-LCU Realization of Circular Convolution}
\label{sec:preliminaries}
\subsection{Classical Circular Convolution}
For two sequences $\bm{a}=(a_0,\dots,a_{N-1})$ and $\bm{b}=(b_0,\dots,b_{N-1})$ of length $N$, the classical circular convolution $\bm{a} \star \bm{b} = \bm{c}=(c_0,\dots,c_{N-1})$ is
\begin{equation}\label{eq:classical_convolution}
    c_k = \sum_{j=0}^{N-1} a_j b_{k-j \bmod N},
    \quad k=0,\dots,N-1.
\end{equation}

If $N$ is a power of $2$, we encode the vector into $n=\log_2 N$ bits. Otherwise, we may embed the vector into the next power-of-two dimension by zero-padding, a standard technique in discrete Fourier transform and convolution algorithms~\cite{Cooley1965,Oppenheim1999}.

\subsection{Group-Theoretic Form}
For a finite group $G$, the convolution of functions $f,g : G\to\mathbb{C}$ is defined as
\begin{equation}
    (f \star g)(u) = \sum_{v\in G} f(uv^{-1}) g(v).
\end{equation}
For the cyclic group $G=\mathbb{Z}/N\mathbb{Z}$, this is equivalent to discrete circular convolution~\cite{Castelazo2022}.
Let $\{L_{i,N}\}_{i=0}^{N-1}$ denote the left-regular representation of $\mathbb{Z}/N\mathbb{Z}$.
These matrices are unitary permutation matrices that act on a vector $\bm{a}$ as 
\begin{equation}
\left( L_{i,N} \bm{a} \right)_k = (\bm{a})_{k-i \pmod N}.
\end{equation}
Furthermore, these representation matrices satisfy the composition law:
\begin{equation} \label{eq:L_combination}
L_{i,N} L_{j,N} = L_{(i+j) \bmod N, \, N}
\end{equation}
Then, the convolution of $\bm{a}$ with the kernel $\bm{b}$ can be written as:
\begin{equation}\label{eq:convolution_with_L}
    \bm{c} 
    = \sum_{i=0}^{N-1} b_i \, L_{i,N} \bm{a}
    =C_N(\bm{b})\bm{a},
\end{equation}

Thus, circular convolution is naturally expressed as a linear combination of group-shift operators. 
This representation forms the starting point for quantum implementations based on the LCU framework.

\subsection{LCU Construction and Modular Adder SELECT}

A general quantum implementation of group convolution was established by Castelazo \textit{et al.}, who formulated convolution over finite groups within the LCU paradigm and then discussed separate Fourier-regime simplifications for the Abelian case~\cite{Castelazo2022}.

Given the decomposition
\begin{equation}
C_N(\bm{b})
=
\sum_{i=0}^{N-1}
b_i\, L_{i,N},
\end{equation}
we formulate the circular convolution construction using two standard ingredients:

\begin{enumerate}
\item A state-preparation unitary $\mathrm{PREP}_b$ that prepares the complex-valued kernel amplitudes
\begin{equation}
\mathrm{PREP}_b \ket{0}^{\otimes n}
=
\ket{\bm{b}}
=
\sum_{i=0}^{N-1} b_i \ket{i},
\qquad
\sum_i |b_i|^2=1,\quad b_i=|b_i|\mathrm{e}^{\mathrm{i}\theta_i},
\end{equation}
(This $\mathrm{PREP}_b$ convention is used throughout this work.
Non-normalized kernels are handled by the scaling discussion in Remark~\ref{rem:ab_normalization}.)

\item A SELECT unitary that applies the group action conditioned on the index register,
\begin{equation}
\mathrm{SELECT}_L :
\ket{i}\ket{k}
\longmapsto
\ket{i}\, L_{i,N}\ket{k}.
\label{eq:select_L}
\end{equation}
\end{enumerate}

To avoid notation clashes with the convolution operator symbol $C_N(\bm{b})$, we denote controlled operations uniformly by
\begin{equation}
W_q[O]
:=
\left(\ket{0}\!\bra{0}\right)_q\!\otimes I
\;+\;
\left(\ket{1}\!\bra{1}\right)_q\!\otimes O,
\label{eq:controlled_W_def}
\end{equation}
where $q$ is the control qubit and $O$ acts on the target register.

The asymmetry of the postselection is essential. A same-state symmetric overlap based on $\mathrm{PREP}_b$ would yield
\begin{equation}
(\bra{0}^{\otimes n}_{A}\otimes I_D)
(\mathrm{PREP}_b^\dagger\otimes I_D)\,
\mathrm{SELECT}_L\,
(\mathrm{PREP}_b\otimes I_D)
(\ket{0}^{\otimes n}_{A}\otimes I_D)
=
\sum_{i=0}^{N-1}|b_i|^2 L_{i,N},
\end{equation}
so the kernel phases cancel and only the weights $|b_i|^2$ remain. We therefore use the asymmetric overlap with left postselection state $\ket{u}$ and right ancilla state $\ket{\bm b}$, which preserves the coefficients $b_i$ themselves.

For the cyclic group $G=\mathbb{Z}/N\mathbb{Z}$, the group action corresponds exactly to modular addition in the computational basis:
\begin{equation}
\ket{i}\ket{k}
\longmapsto
\ket{i}\ket{k+i \!\!\!\pmod{N}}.
\end{equation}

Hence, the SELECT operation can be instantiated using standard quantum modular adder circuits~\cite{Wang2017}. 
Well-known constructions include ripple-carry adders (such as the Vedral-Barenco-Ekert (VBE) adder~\cite{Vedral1996} and Cuccaro adders~\cite{Cuccaro2004}) and QFT-based adders~\cite{Draper2000}.
Ripple-carry adders (Cuccaro architecture) admit linear-size/depth implementations in $n$, whereas exact QFT-based adders use $\mathcal{O}(n^2)$ controlled phases under a straightforward logical-gate count, up to architecture-dependent depth/precision trade-offs.
With these ingredients, in the circular-convolution setting considered here, $C_N(\bm{b})$ is obtained as the postselected block of
$(\mathrm{PREP}_u^\dagger\!\otimes I)\,\mathrm{SELECT}_L\,(\mathrm{PREP}_b\!\otimes I)$,
in the standard LCU/block-encoding sense~\cite{Childs2012,Berry2015,Gilyen2019},
as stated precisely in Theorem~\ref{thm:asym_modadder_conv}.
For the remainder of this paper, whenever $N=2^n$, we use the equivalent $n$-qubit notation
\begin{equation}
L_{i,n}:=L_{i,N}.
\label{eq:L_n_from_L_N}
\end{equation}
\begin{theorem}[Asymmetric LCU realization of circular convolution]
\label{thm:asym_modadder_conv}
Let $N=2^n$, with ancilla/index register $A$ and data register $D$, each of size $n$ qubits.
Define the controlled-shift operator:
\begin{equation}
L_{i,n}\ket{k}=\ket{k+i \pmod N},\qquad
\mathrm{SELECT}_{L}
:=\sum_{i=0}^{N-1}\ket{i}\!\bra{i}_{A}\otimes L_{i,n}.
\label{eq:select_L_def}
\end{equation}
Let $\ket{\bm{a}}_D = \sum_{k=0}^{N-1}a_k\ket{k}_D$ be an arbitrary data state. 
Assume unitaries $\mathrm{PREP}_b$ and $\mathrm{PREP}_u$ prepare the kernel state $\ket{\bm{b}}_A = \sum_{i=0}^{N-1}b_i\ket{i}_A$ and the uniform state $\ket{u}_A = \frac{1}{\sqrt{N}}\sum_{i=0}^{N-1}\ket{i}_A$ from $\ket{0}^{\otimes n}_A$, respectively. 
Consider the asymmetric LCU operator
\begin{equation}
\mathcal U_{\mathrm{LCU}}^{(L)}
:=
(\mathrm{PREP}_u^\dagger\otimes I_{D})\,
\mathrm{SELECT}_{L}\,
(\mathrm{PREP}_b\otimes I_{D}).
\label{eq:U_LCU_L}
\end{equation}
Then, projecting the ancilla register to $\bra{0}^{\otimes n}_A$ yields
\begin{equation}
(\bra{0}^{\otimes n}_{A}\otimes I_{D})\,
\mathcal U_{\mathrm{LCU}}^{(L)}\,
(\ket{0}^{\otimes n}_{A}\otimes\ket{\bm{a}}_{D})
=
\frac{1}{\sqrt N}\,C_n(\bm{b})\ket{\bm{a}}_{D},
\label{eq:block_of_Cn}
\end{equation}
where $C_n(\bm{b}):=\sum_{i=0}^{N-1} b_i\,L_{i,n}$.
Hence the output amplitudes satisfy
\begin{equation}
\big(C_n(\bm{b})\bm{a}\big)_y
=
\sum_{i=0}^{N-1} b_i\,a_{(y-i)\bmod N},
\end{equation}
which is the discrete circular convolution.
\end{theorem}

\begin{proof}
Using
\begin{equation}
(\bra{0}^{\otimes n}\mathrm{PREP}_u^\dagger)=\bra{u},
\end{equation}
we obtain
\begin{align}
&(\bra{0}^{\otimes n}_{A}\otimes I_{D})\,
\mathcal U_{\mathrm{LCU}}^{(L)}\,
(\ket{0}^{\otimes n}_{A}\otimes\ket{\bm{a}}_{D}) \nonumber \\
&= (\bra{u}_{A}\otimes I_{D})\,
\mathrm{SELECT}_{L}\,
(\ket{\bm{b}}_{A}\otimes\ket{\bm{a}}_{D}) \nonumber \\
&= (\bra{u}_{A}\otimes I_{D})\,
\sum_{i=0}^{N-1}b_i\ket{i}_{A}\otimes L_{i,n} \ket{\bm{a}}_{D} \nonumber \\
&= \sum_{i=0}^{N-1}b_i \langle u|i\rangle_{A}
\, L_{i,n} \ket{\bm{a}}_{D} \nonumber \\
&= \frac{1}{\sqrt N}\sum_{i=0}^{N-1} b_i\,L_{i,n}\ket{\bm{a}}_{D}
= \frac{1}{\sqrt N}\,C_n(\bm{b})\ket{\bm{a}}_{D}.
\end{align}
Expanding in the computational basis gives
\begin{equation}
\big(C_n(\bm{b})\bm{a}\big)_y
=
\sum_{i=0}^{N-1} b_i\,a_{(y-i)\bmod N},
\end{equation}
i.e., circular convolution.
\end{proof}

\begin{remark}[Asymmetry and the state-input specialization]
\label{rem:asym_state_input}
Theorem~\ref{thm:asym_modadder_conv} is written in the standard oracle language with $\mathrm{PREP}_b$. The practical advantage of the asymmetric specialization appears when the kernel state is already available externally. In that setting one may use
\begin{equation}
\mathcal V_{\mathrm{state}}^{(L)}:=
(\mathrm{PREP}_u^\dagger\otimes I_D)\,\mathrm{SELECT}_L,
\end{equation}
so that
\begin{equation}
(\bra{0}^{\otimes n}_{A}\otimes I_D)\,
\mathcal V_{\mathrm{state}}^{(L)}\,
(\ket{\bm b}_{A}\otimes\ket{\bm a}_{D})
=
\frac{1}{\sqrt N}\,C_n(\bm{b})\ket{\bm a}_{D}.
\end{equation}
Thus the convolution subroutine itself requires only the fixed uncompute $\mathrm{PREP}_u^\dagger$ and no kernel-dependent inverse preparation $\mathrm{PREP}_b^\dagger$. This same asymmetry is also what preserves the complex coefficients $b_i$; a same-state symmetric overlap using $\mathrm{PREP}_b$ on both sides would instead give $\sum_i |b_i|^2 L_{i,n}$.
\end{remark}

\begin{remark}[Normalization of $\bm{a}$ and $\bm{b}$]
\label{rem:ab_normalization}
Amplitude encoding requires normalized states.  
For arbitrary nonzero classical vectors $\bm{a},\bm{b}\in\mathbb C^N$, define
\begin{equation}
\hat{\bm{a}}:=\frac{\bm{a}}{\|\bm{a}\|_2},\qquad
\hat{\bm{b}}:=\frac{\bm{b}}{\|\bm{b}\|_2},
\end{equation}
and prepare
\begin{equation}
\mathrm{PREP}_{a}\ket{0}^{\otimes n}=\mathrm{PREP}_{D}\ket{0}^{\otimes n}=\ket{\hat{\bm{a}}},\qquad
\mathrm{PREP}_{b}\ket{0}^{\otimes n}=\mathrm{PREP}_{A}\ket{0}^{\otimes n}=|\hat{\bm{b}}\rangle.
\end{equation}
Then Eq.~\eqref{eq:block_of_Cn} becomes
\begin{equation}
(\bra{0}^{\otimes n}_{A}\otimes I_{D})\,
\mathcal U_{\mathrm{LCU}}^{(L)}\,
(\ket{0}^{\otimes n}_{A}\otimes\ket{\hat{\bm{a}}}_{D})
=
\frac{1}{\sqrt N}\,C_n(\hat{\bm{b}})\ket{\hat{\bm{a}}}_{D}
=
\frac{1}{\sqrt N\,\|\bm{a}\|_2\|\bm{b}\|_2}\sum_{y=0}^{N-1}(C_n(\bm{b})\bm{a})_y\ket{y}_{D}.
\end{equation}
Hence the postselected data state is proportional to the state vector associated with the unnormalized circular convolution
$C_n(\bm{b})\bm{a}$, with a known global scale factor
$\big(\sqrt N\,\|\bm{a}\|_2\|\bm{b}\|_2\big)^{-1}$.
The LCU success probability is
\begin{equation}
p_{\mathrm{succ}}
=
\frac{\|C_n(\bm{b})\bm{a}\|_2^2}
{N\,\|\bm{a}\|_2^2\,\|\bm{b}\|_2^2}.
\end{equation}
(If $\bm{a}=\bm{0}$ or $\bm{b}=\bm{0}$, the convolution output is trivially $\bm{0}$.)
\end{remark}
Theorem~\ref{thm:asym_modadder_conv} gives the operator-level asymmetric-LCU formulation of circular convolution. Lemma~\ref{lem:adder_independence} then shows that standard modular addition realizes the required unitary $\mathrm{SELECT}_L$.

\begin{lemma}[Modular adder realization of $\mathrm{SELECT}_L$]
\label{lem:adder_independence}
Using the definitions of $L_{i,n}$ and $\mathrm{SELECT}_L$ from Eq.~\eqref{eq:select_L_def}, define the modular addition operator on the computational basis by
\begin{equation}
\mathrm{ADD}_N\ket{i}_{A}\ket{k}_{D}:=\ket{i}_{A}\ket{k+i \!\!\!\pmod N}_{D}.
\end{equation}
Then
\begin{equation}
\mathrm{ADD}_N=\mathrm{SELECT}_L.
\end{equation}
Hence any exact circuit synthesis of $\mathrm{ADD}_N$ (e.g., ripple-carry adder or QFT adder) yields the same $\mathrm{SELECT}_L$, and therefore the same convolution block in Eq.~\eqref{eq:U_LCU_L}.
\end{lemma}

\begin{proof}
For any computational-basis input $\ket{i}_{A}\ket{k}_{D}$,
\begin{align}
\mathrm{SELECT}_L(\ket{i}_{A}\ket{k}_{D})
&=
\left(\sum_{j=0}^{N-1}\ket{j}\!\bra{j}_{A}\otimes L_{j,n}\right)
(\ket{i}_{A}\ket{k}_{D}) \\
&=
\sum_{j=0}^{N-1}\left(\ket{j}\!\braket{j}{i}_{A}\right)\otimes L_{j,n}\ket{k}_{D} \\
&=
\ket{i}_{A}\otimes L_{i,n}\ket{k}_{D} \\
&=
\ket{i}_{A}\ket{k+i \!\!\!\pmod N}_{D} \\
&=
\mathrm{ADD}_N(\ket{i}_{A}\ket{k}_{D}).
\end{align}
Thus $\mathrm{SELECT}_L$ and $\mathrm{ADD}_N$ coincide on the computational basis, so they are the same unitary.
The LCU statement follows by direct substitution into $\mathcal U_{\mathrm{LCU}}^{(L)}$.
\end{proof}

\section{Symmetrized Convolution and Recursive Operator Structure}\label{sec:quantum_representation_of_convolution}
We now introduce the reversal matrix $J_N$ and the symmetrized convolution operator used throughout the remainder of this work. 
Define the reversed input $\bm{a}_{\mathrm{R}}:=J_N\bm{a}$, i.e.,
\begin{equation}
(\bm{a}_{\mathrm{R}})_i=(\bm{a})_{N-1-i}.
\end{equation}
Then we define
\begin{equation}
    H_N(\bm{b}) := \sum_{i=0}^{N-1} b_i \, L_{i,N} J_{N},  \label{eq:LCU_form_of_convolution}
\end{equation}
so that
\begin{equation}
\bm{c}=C_N(\bm{b})\bm{a}=H_N(\bm{b})\bm{a}_{\mathrm{R}}.
\end{equation}
Equivalently, since $J_N^2=I$,
\begin{equation}\label{eq:C_from_H}
C_N(\bm{b})=H_N(\bm{b})J_N.
\end{equation}
This representation is equivalent to standard circular convolution, but it has two technical advantages for our purpose. The introduction of the reversal matrix $J_N$ plays a crucial role not only in enabling the recursive circuit construction discussed next, but also in ensuring the Hermiticity of the operator for real-valued kernels. As we will show later in Sec.~\ref{sec:blockencoding_qsvt}, this Hermiticity provides a direct Hermitian route for QSVT.

Henceforth, we restrict the discussion to dimensions of the form $N = 2^n$ and identify the vector space $\mathbb{C}^N$ with the $n$-qubit Hilbert space $(\mathbb{C}^2)^{\otimes n}$. In this notation we write
\begin{equation}
J_n := J_N,
\end{equation}
so Eq.~\eqref{eq:C_from_H} becomes
\begin{equation}
C_n(\bm{b})=H_n(\bm{b})J_n.
\end{equation}
By definition, the action of the left-regular representation $L_{i,n}$ on the computational basis states $\{\ket{k}\}_{k=0}^{N-1}$ is given by
\begin{equation}\label{eq:L_action_on_ket}
    L_{i,n} \ket{k} = \ket{k+i \pmod N}
\end{equation}

\subsection{Reversal Operators}
A reversal operator $J_n$ is defined by its action on the computational basis:
\begin{equation}\label{eq:J_action_on_ket}
    J_n \ket{s} := \ket{2^n-1-s}, \quad \text{for } s=0,\dots,2^n-1.
\end{equation}
Equivalently, in matrix form this operator is the anti-identity permutation; under the little-endian convention used here, it is implemented simply as a bitwise NOT layer.
Throughout this work, we use little-endian integer indexing:
\begin{equation}
k=\sum_{j=0}^{n-1}k_j2^j,\qquad k_j\in\{0,1\},
\end{equation}
where $k_0$ is the least-significant bit (LSB) and the basis label is written as $\ket{k_{n-1}\cdots k_0}$.
The action of $J_n$ on this state is
\begin{equation}
J_n\ket{k_{n-1}\cdots k_0}
=\ket{2^n-1-k}
= X^{\otimes n}\ket{k_{n-1}\cdots k_0}.
\end{equation}
Consequently, the reversal matrix can be expressed as the $n$-fold tensor product of the Pauli-$X$ gate,
\begin{equation}
J_n = X^{\otimes n},
\end{equation}
and therefore $J_n^2=I$.
For convenience, we introduce the reflected shift operator as
\begin{equation}
     \widetilde{L}_{i,n} := L_{i,n} J_n.
\end{equation}
Substituting this into Eq.~\eqref{eq:LCU_form_of_convolution}, the symmetrized convolution operator becomes
\begin{equation}
     H_n(\bm{b}) = \sum_{i=0}^{2^n-1} b_i \, \widetilde{L}_{i,n}.\label{eq:convolution_operator}
\end{equation}

\subsection{Reflected Generator and Its Structural Recursion}
\label{sec:recursion}
We introduce the reflected generator
\begin{equation}\label{eq:u_n}
    U_n := L_{1,n}J_n = \widetilde{L}_{1,n}.
\end{equation}
Its role is structural: it provides a compact normal form for the $J_n$-symmetrized reflected-shift family, rather than the preferred arithmetic implementation of the unit incrementer $L_{1,n}$.

\begin{proposition}[Structural recursion of $U_n$]
\label{prop:Un-recursion}
The reflected generator satisfies
\begin{equation}
    U_{n+1}
    =
    U_n\otimes\ket{0}\!\bra{0}
    +
    J_n\otimes\ket{1}\!\bra{1},
    \label{eq:Un-recursion}
\end{equation}
with base case $U_1=I$.
\end{proposition}

The proof and expanded projector form of $U_n$ are deferred to Appendix~\ref{app:Un_recursion}. This recursion should be interpreted as an operator-level structural relation. The efficient compiled realization used later instead follows the standard binary carry-propagation recursion of the increment operator.
Another consequence of this structural recursion is a sharp constraint on the Pauli support of the reflected-shift family; see Appendix~\ref{app:pauli_support}.

Using $L_{1,n}=U_nJ_n$, higher shifts obey
\begin{equation}
    L_{i,n}=(U_nJ_n)^i,
    \qquad
    \widetilde{L}_{i,n}=(U_nJ_n)^iJ_n.
\end{equation}
For binary synthesis we use
\begin{equation}
    L_{2^m,n}
    =
    L_{1,n-m}\otimes I^{\otimes m}
    =
    \left(U_{n-m}J_{n-m}\right)\otimes I^{\otimes m},
\label{eq:binary_synthesis}
\end{equation}
together with the binary expansion $i=\sum_{m=0}^{n-1} i_m2^m$, which underlies the circuit synthesis developed later.

\subsection{Optional Pauli-String Baseline}
\label{subsec:pauli_brief}
For completeness, one may also expand the reflected shifts $\widetilde{L}_{i,n}$ into Pauli strings and apply standard Pauli-LCU synthesis. This route is compatible with the operator structure developed in Sec.~\ref{sec:quantum_representation_of_convolution}, but it is not the focus of this work. The reflected-shift family has Pauli support contained in $\{I,X\}\otimes\{I,X,Z\}^{\otimes(n-1)}$, so no Pauli-$Y$ terms appear and the support size is at most $2\cdot 3^{n-1}$; see Appendix~\ref{app:pauli_support}. Even with this sharper bound, the support still grows exponentially with $n$, so coefficient enumeration and preprocessing can dominate the end-to-end cost. Accordingly, the main text emphasizes the direct modular adder compilation of the $\mathrm{SELECT}_{\widetilde{L}}$ block and uses Pauli-LCU only as an optional reference baseline.

\section{Quantum Circuit Construction}
\label{sec:circuits}

Following Theorem~\ref{thm:asym_modadder_conv}, this section turns the modular adder/SELECT equivalence into an explicit convolution circuit. We first present the direct recursive realization of $\mathrm{SELECT}_{\widetilde{L}}$ as a structural normal form that makes the role of $J_n$ and the reflected-generator algebra transparent. This main-text recursion is explanatory rather than gate-optimal. The stronger practical resource bounds used later come from the exactly equivalent compiled carry-propagation realization collected in Appendix~\ref{app:compiled_suffix_recursion}. We then show how standard QFT and ripple-carry adders realize the same underlying $\mathrm{SELECT}_L$ block.

\subsection{Implementation via Explicit Recursion of the Reflected Shift Operator}

The reflected generator $U_n$ is introduced in Sec.~\ref{sec:recursion}, while its detailed recursive derivation is given in Appendix~\ref{app:Un_recursion}. From Eq.~\eqref{eq:u_n}, we have $L_{1,n}=U_nJ_n$.
Operationally, this product acts as the unit cyclic incrementer, and powers of this block generate all shifts needed for the convolution, i.e., $L_{i,n}=(L_{1,n})^i$.
The full convolution operation is realized using the LCU framework. 
For a system of $n$ data qubits (dimension $N=2^n$), we introduce $n$ ancilla qubits to encode the kernel $\ket{\bm{b}}$.
The convolution operator is constructed as a superposition of the reflected shift operators as in Eq.~\eqref{eq:convolution_operator}:
\begin{equation}
    H_n(\bm{b}) = \sum_{i=0}^{2^n-1} b_i\, \widetilde{L}_{i,n} = \sum_{i=0}^{2^n-1} b_i \left( U_n J_n \right)^i J_n.
\label{eq:convolution_op}
\end{equation}
A literal implementation of Eq.~\eqref{eq:convolution_op} would repeat the unit incrementer $i$ times, requiring $\mathcal{O}(i\cdot \mathrm{poly}(n))$ gates in the worst case.

\begin{remark}[Polynomial synthesis of $\widetilde{L}_{i,n}$ within the recursive framework]
\label{rem:poly_select_construction}
To keep the construction in polynomial complexity of $n$, we write
\begin{equation}
i=\sum_{m=0}^{n-1} i_m 2^m,\qquad i_m\in\{0,1\}.
\end{equation}
For each bit position $m$, define the block shift
\begin{equation}
L_{2^m,n}=L_{1,n-m}\otimes I^{\otimes m}
=\left(U_{n-m}J_{n-m}\right)\otimes I^{\otimes m},
\label{eq:pow2_shift}
\end{equation}
which increments the subregister of size $(n-m)$ and leaves the $m$ LSB unchanged. Then
\begin{equation}
L_{i,n}
=\prod_{m=0}^{n-1}\left[L_{2^m,n}\right]^{i_m}
=\prod_{m=0}^{n-1}\left[\left(U_{n-m}J_{n-m}\right)\otimes I^{\otimes m}\right]^{i_m},
\label{eq:binary_shift_decomp}
\end{equation}
yielding the synthesis of the reflected shift operator:
\begin{equation}
  \widetilde{L}_{i,n}= \Big(\Pi_{m=0}^{n-1} [L_{2^m,n}]^{i_m} \Big)J_n.
\end{equation}

\end{remark}

Using the binary decomposition in Eq.~\eqref{eq:binary_shift_decomp}, we define the reflected-shift block by
\begin{equation}
\mathrm{SELECT}_{\widetilde{L}}
:=
\left(\prod_{m=0}^{n-1}W_{A_m}[L_{2^m,n}]\right)(I_A\otimes J_n),
\label{eq:select_tildeL_coherent}
\end{equation}
where we use the controlled-operation notation $W_q[\cdot]$ from Eq.~\eqref{eq:controlled_W_def}. 
The corresponding LCU unitary is then
\begin{equation}
\mathcal{U}_{\mathrm{LCU}}^{(\widetilde{L})}
=
(\mathrm{PREP}_u^\dagger\otimes I)\,\mathrm{SELECT}_{\widetilde{L}}\,(\mathrm{PREP}_b\otimes I).
\label{eq:coherent_select_pipeline}
\end{equation}
Notice that the product of controlled block shifts physically synthesizes the abstract $\mathrm{SELECT}_L$ operator defined in Theorem~\ref{thm:asym_modadder_conv}.
Thus, we have the exact operator equivalence:
\begin{equation}
\mathrm{SELECT}_{\widetilde{L}}
=\mathrm{SELECT}_{L}\,(I_{A}\otimes J_n),
 \label{eq:select_bridge}
\end{equation}
where subscript $A$ denotes the ancilla/index register. 
Thus the two SELECT operators in Eq.~\eqref{eq:select_bridge} differ only by one data-register layer $J_n=X^{\otimes n}$.
In the application of $\mathrm{SELECT}_{\widetilde{L}}$, Eq.~\eqref{eq:select_tildeL_coherent} indicates $J_n$ first acts on the data register, followed by the controlled block shifts $W_A$.
From Eq.~\eqref{eq:C_from_H}, the compiled synthesis for $H_n(\bm{b})$ also implements the standard convolution operator by one additional input-side $J_n$, equivalently by preparing $\ket{\bm{a}_{\mathrm{R}}}=J_n\ket{\bm{a}}$ as the data input.
Hence the central $\mathrm{SELECT}_{\widetilde{L}}$ is synthesized as a single unitary block using resources polynomial in $n$, making the LCU operator structure explicit without instantiating the $N$ reflected shifts one by one (see Remark~\ref{rem:poly_select_construction}).

Figure~\ref{fig:qc_lcu} summarizes the overall reflected-shift convolution pipeline. Its only realization-dependent ingredient is the controlled block shift $W_{A_m}[L_{2^m,n}]$. The underlying reflected-generator recursion from which the direct realization is derived is given in Appendix~\ref{app:Un_recursion}, while the two explicit recursive realizations of the block itself are recorded in Appendix~\ref{app:shiftblock_recursions}.

\begin{algorithm}[H]
    \caption{Quantum convolution via $\mathrm{SELECT}_{\widetilde{L}}$}
    \label{alg:quantum_conv}
    \begin{algorithmic}[1]
        \renewcommand{\algorithmicrequire}{\textbf{Require:}}
        \renewcommand{\algorithmicensure}{\textbf{Ensure:}}

        \Require Registers $D$ and $A$, each with $n$ qubits; $N=2^n$
        \Statex \hspace{\algorithmicindent}Unitaries $\mathrm{PREP}_{a},\mathrm{PREP}_b,\mathrm{PREP}_u$
        \Ensure Upon postselection $A=\ket{0}^{\otimes n}$, register $D$ is proportional to $H_n(\bm{b})\ket{\bm{a}}=C_n(\bm{b})\ket{\bm{a}_{\mathrm{R}}}$

        \State Initialize $A\leftarrow\ket{0}^{\otimes n}$ and $D\leftarrow\ket{0}^{\otimes n}$
        \State Apply $\mathrm{PREP}_{a}$ on $D$ and $\mathrm{PREP}_b$ on $A$
        \State Apply $J_n$ on $D$
        \For{$m = 0$ \textbf{to} $n-1$}
            \State Apply $W_{A_m}[L_{2^m,n}]$ on $D$
        \EndFor
        \State Apply $\mathrm{PREP}_u^\dagger$ on $A$
        \State Measure $A$
        \If{outcome is $\ket{0}^{\otimes n}$}
            \State \textbf{Success:} output register $D$
        \Else
            \State \textbf{Failure:} repeat (or use amplitude amplification)
        \EndIf
    \end{algorithmic}
\end{algorithm}

\begin{figure}[htbp]
\centering
\resizebox{\linewidth}{!}{\input{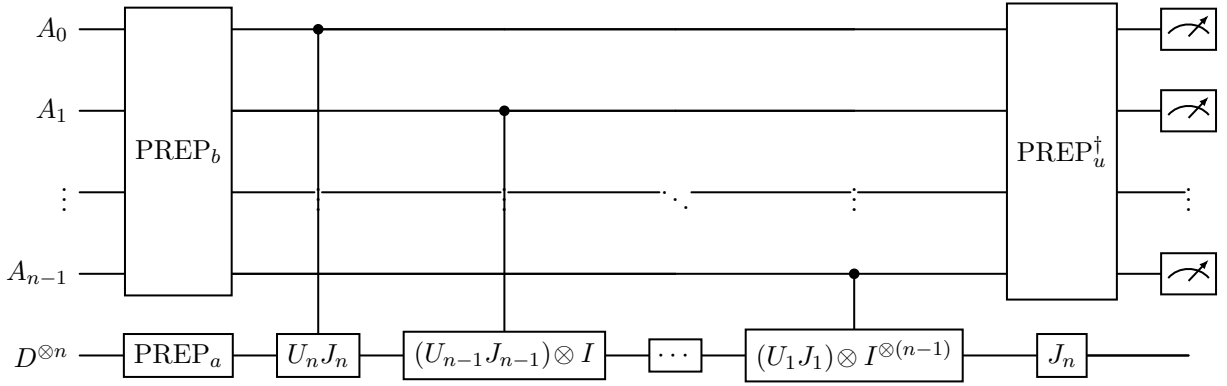}}
\caption{General $n$-qubit recursive LCU circuit (schematic). The middle block is a $\mathrm{SELECT}_{\widetilde{L}}$ in the $U_n$-nesting form; the first and last controlled blocks are shown explicitly, and $\cdots$ denotes the intermediate ones. The ancilla is uncomputed by applying $\mathrm{PREP}_u^\dagger$.}
\label{fig:qc_lcu}
\end{figure}

\subsection{QFT-Adder and Ripple-Carry Adder Realizations: Equivalence and Trade-Offs}
\label{subsec:qft_equiv}

Standard quantum adder libraries are typically designed to implement the standard shift $\mathrm{SELECT}_L$.
According to Eq.~\eqref{eq:select_bridge}, any such synthesis can be directly converted into an exact realization of $\mathrm{SELECT}_{\widetilde{L}}$ by composition with $(I_A\otimes J_n)$. 
In terms of quantum circuits, this simply requires applying the $J_n$ operator to the data register prior to the $\mathrm{SELECT}_L$ sequence.
Therefore, implementing $\mathrm{SELECT}_{\widetilde{L}}$ is equivalent to implementing standard modular addition, up to the negligible cost of one layer of $X$ gates.

Following Lemma~\ref{lem:adder_independence} and Theorem~\ref{thm:asym_modadder_conv}, any exact quantum circuit for $\mathrm{ADD}_N=\mathrm{SELECT}_L$ provides a valid realization of the standard shift block. Consequently, via Eq.~\eqref{eq:select_bridge}, it also synthesizes our target operator $\mathrm{SELECT}_{\widetilde{L}}$.
To illustrate this modularity, we present two canonical instantiations using a QFT adder (Fig.~\ref{fig:qc_qft_adder_conv}) and a ripple-carry adder (Fig.~\ref{fig:qc_ripple_adder_conv}).

\myheading{QFT adder implementation}
The exact $\mathrm{SELECT}_L$ operator established in Theorem~\ref{thm:asym_modadder_conv} can also be synthesized in the Fourier basis.
It is important to emphasize that the QFT is employed here solely to realize the modular addition function within the LCU $\mathrm{SELECT}_L$ block.
The overall algorithm remains a real-domain LCU construction, which is fundamentally distinct from implementing the convolution itself via the Fourier-domain convolution theorem.

Let $N=2^n$, $\omega_N:=\mathrm{e}^{2\pi \mathrm{i}/N}$.
The QFT over $\mathbb{Z}/N\mathbb{Z}$ is defined as:
\begin{equation}
\mathrm{QFT}_N\ket{k}
=\frac{1}{\sqrt N}\sum_{t=0}^{N-1}\omega_N^{kt}\ket{t}.
\end{equation}
In Fourier basis, each cyclic shift is diagonalized as
\begin{equation}
L_{i,n}
=
\mathrm{QFT}_N^\dagger ~D_i ~\mathrm{QFT}_N,
\qquad
D_i:=\sum_{t=0}^{N-1}\omega_N^{it}\ket{t}\!\bra{t}.
\end{equation}
Hence $\mathrm{SELECT}_L$ can be written as
\begin{equation}
\mathrm{SELECT}_L
=
(I_A\otimes \mathrm{QFT}_N^\dagger)\,
\Phi_{A\to D}\,
(I_A\otimes \mathrm{QFT}_N),
\end{equation}
where $\Phi_{A\to D}$ is the controlled diagonal phase network
\begin{equation}
\Phi_{A\to D}
=
\sum_{i=0}^{N-1}\ket{i}\!\bra{i}_A\otimes D_i
=
\sum_{i=0}^{N-1}\sum_{t=0}^{N-1}\omega_N^{it}\ket{i}\!\bra{i}_A\otimes\ket{t}\!\bra{t}_D,
\end{equation}
equivalently satisfying
\begin{equation}
\Phi_{A\to D}\ket{i}_A\ket{t}_D
=
\omega_N^{it}\ket{i}_A\ket{t}_D.
\end{equation}
This sequence exactly yields
\begin{equation}
\ket{i}\ket{k}\mapsto\ket{i}\ket{k+i \!\!\!\pmod N},
\end{equation}
confirming that the QFT-based construction is an exact realization of $\mathrm{SELECT}_L$.

\begin{figure}[t]
\centering
\resizebox{\linewidth}{!}{%
\input{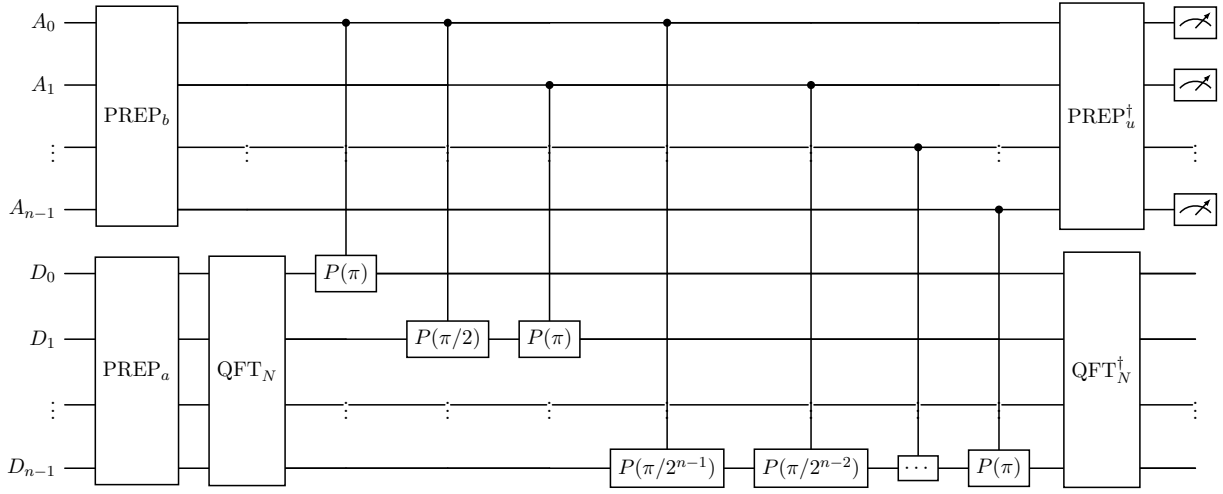}
}
\caption{General $n$-qubit asymmetric-LCU convolution circuit with a QFT adder realization of $\mathrm{SELECT}_L$. The middle block is implemented as $\mathrm{QFT}_N\rightarrow \Phi_{A\to D}\rightarrow \mathrm{QFT}_N^\dagger$ on the data register $D$, with control from ancilla/index register $A$.}
\label{fig:qc_qft_adder_conv}
\end{figure}

\myheading{Ripple-carry implementation}
Alternatively, $\mathrm{SELECT}_L$ can be implemented directly with a ripple-carry modular adder:
\begin{equation}
\mathrm{ADD}^{\mathrm{rc}}_N\ket{i}\ket{k}
=
\ket{i}\ket{k+i \!\!\!\pmod N}.
\end{equation}
By employing standard exact ripple-carry architectures (e.g., Cuccaro/Takahashi variants~\cite{Cuccaro2004,Takahashi2010}), the adder core uses linear gate count/depth.
Under our abstract logical-gate model, this yields an adder complexity of $\mathcal{O}(n)=\mathcal{O}(\log N)$.
Consequently, replacing the central $\mathrm{SELECT}_L$ block by a ripple-carry modular adder synthesis yields an LCU convolution circuit with an adder complexity that scales strictly linearly with $n$. 

\begin{figure}[t]
\centering
\resizebox{0.93\linewidth}{!}{%
\input{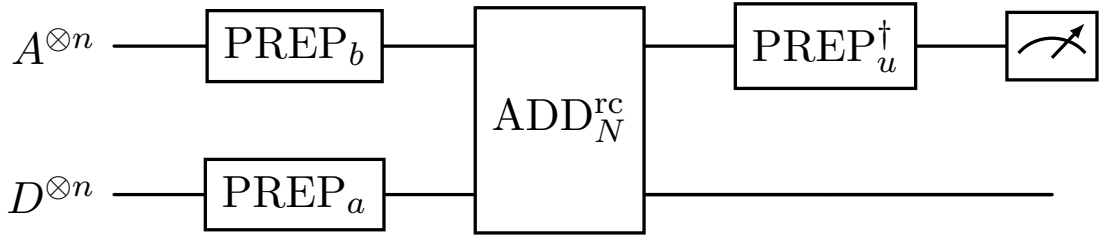}
}
\caption{Asymmetric-LCU convolution circuit with a ripple-carry adder realization of $\mathrm{SELECT}_L$. The middle block implements $\mathrm{ADD}^{\mathrm{rc}}_N:(i,k)\mapsto(i,k+i \bmod N)$. Ancillary carry wires required by a specific ripple-carry variant are omitted for schematic clarity.}
\label{fig:qc_ripple_adder_conv}
\end{figure}

\paragraph{Equivalence and adder choice.}
Both QFT and ripple-carry constructions realize the identical $\mathrm{SELECT}_L$ core.
As established in Eq.~\eqref{eq:select_bridge}, this core maps to the target $\mathrm{SELECT}_{\widetilde{L}}$ via composition with $(I_{A}\otimes J_n)$, which physically corresponds to applying $J_n$ to the data register prior to the adder sequence. 
Consequently, QFT adder, ripple-carry adder, and bitwise recursive compilation are interchangeable, differing only in their compilation trade-offs.
This modularity allows one to select any modular adder optimized for a specific hardware architecture.
Recent advances in this area include logarithmic-depth lookahead-style adders and sublinear-depth carry-save modular addition architectures~\cite{Gaur2024LogDepthAdder,Mitra2025QCSA,Remaud2025AncillaFree,Wang2025}.
In this work, we keep QFT adder and ripple-carry adder as canonical examples.

\section{Complexity Analysis}
\label{sec:complexity}
The complexity discussion below distinguishes between two recursive realizations of the same target unitary: the direct recursive construction used in the main text to expose the operator structure, and the exactly equivalent gate-optimized bitwise recursive compilation constructed in Appendix~\ref{app:compiled_suffix_recursion} and identified block-by-block with the same target in Proposition~\ref{prop:exact_equiv}. The direct recursive construction is structurally transparent but not gate-optimal, whereas the compiled bitwise recursive realization preserves the same unitary action while providing a more efficient gate-level realization.

We evaluate the performance of the direct recursive construction and the compiled bitwise recursive realization in comparison with classical fast Fourier transform (FFT)-based convolution algorithms, as well as the quantum LCU and Fourier-domain approaches discussed by Castelazo \textit{et al.}~\cite{Castelazo2022}. Here, $N = 2^n$ denotes the number of data points. Throughout this section, the relevant distinction is the input model. If the data and kernel are already available as coherent quantum states, the relevant cost is the circuit cost of the convolution block itself. If they are classical vectors, one must additionally include the state-preparation/loading overhead, denoted by $T_{\mathrm{load}}(N)$, which may dominate the total runtime.

\begin{remark}[Complexity of the direct recursive realization]
\label{rem:recursive_direct_complexity}
The main text recursion is primarily a structural construction. It gives a direct recursive synthesis of the reflected-shift family, but it is not intended as the most gate-efficient compilation. The larger cost of this direct realization is not caused by the single layer $J_n=X^{\otimes n}$ itself, whose gate cost is only linear in $n$. Rather, it comes from expanding the structural recursion of $U_n$ under nested controls: the recursion repeatedly generates pattern-controlled $J_i$ layers, and these controlled expansions dominate the macro block and primitive-gate counts.

Consider the block shift $L_{2^m,n}=L_{1,n-m}\otimes I^{\otimes m}$ and let $k:=n-m$ denote the size of the active suffix register. In the direct recursive realization, the expanded projector form of $U_k$ given in Appendix~\ref{app:Un_recursion} contains, for each $i=1,\dots,k-1$, a pattern-controlled layer $J_i=X^{\otimes i}$.
Thus the $i$-th recursive layer contributes $\mathcal{O}(i)$ macro blocks, and the total macro block cost of the block shift is
\begin{equation}
T_{\mathrm{rec}}^{\mathrm{blk}}(L_{2^m,n})
=
\sum_{i=1}^{k-1}\mathcal{O}(i)+\mathcal{O}(k)
=
\mathcal{O}(k^2)
=
\mathcal{O}((n-m)^2),
\end{equation}
where the final $\mathcal{O}(k)$ term comes from the $J_k$ layer in $L_{1,k}=U_kJ_k$.

If each controlled- or multi-controlled-$X$ gate with $r$ controls is further decomposed into primitive CNOT gates using a standard ancilla-assisted synthesis with $\mathcal{O}(r)$ CNOT cost, then the contribution of the $i$th recursive layer to the primitive CNOT gate count is $\mathcal{O}(i(k-i))$~\cite{Claudon2024,Rosa2025,Zindorf2025}.
Therefore
\begin{equation}
T_{\mathrm{rec}}^{\mathrm{CNOT}}(L_{2^m,n})
=
\sum_{i=1}^{k-1}\mathcal{O}(i(k-i))+\mathcal{O}(k)
=
\mathcal{O}(k^3)
=
\mathcal{O}((n-m)^3).
\end{equation}

Summing over all suffix sizes appearing in Eq.~\eqref{eq:select_tildeL_coherent} yields
\begin{equation}
T_{\mathrm{rec}}^{\mathrm{blk}}(\mathrm{SELECT}_{\widetilde{L}})
=
\sum_{k=1}^{n}\mathcal{O}(k^2)
=
\mathcal{O}(n^3),
\end{equation}
and
\begin{equation}
T_{\mathrm{rec}}^{\mathrm{CNOT}}(\mathrm{SELECT}_{\widetilde{L}})
=
\sum_{k=1}^{n}\mathcal{O}(k^3)
=
\mathcal{O}(n^4).
\end{equation}
Thus the direct recursive realization is cubic in the macro block model and quartic in primitive CNOT count.
\end{remark}

\begin{remark}[Equivalent optimized bitwise recursive compilation]
\label{rem:compiled_optimized_complexity}
Appendix~\ref{app:compiled_suffix_recursion} constructs the same reflected-shift unitary in a bitwise recursive carry-propagation form, and Proposition~\ref{prop:exact_equiv} proves that this compilation is exactly equivalent to the main text construction. This optimized realization should therefore be viewed not as a different algorithm, but as a more efficient gate-level realization of the same recursive construction.
In this sense, the optimized bitwise recursive compilation removes the overhead associated with explicitly expanding those nested controlled $J_i$ layers, while preserving exactly the same reflected-shift block.

For a suffix block of size $k=n-m$, the optimized compiled realization uses $\mathcal{O}(k)$ macro blocks. Under the same standard ancilla-assisted decomposition~\cite{Claudon2024,Rosa2025,Zindorf2025} with $\mathcal{O}(r)$ primitive CNOT cost for an $r$-controlled-$X$, its primitive CNOT cost is
\begin{equation}
\sum_{r=1}^{k}\mathcal{O}(r)=\mathcal{O}(k^2).
\end{equation}
Summing over all suffix sizes therefore gives
\begin{equation}
T_{\mathrm{cmp}}^{\mathrm{blk}}(\mathrm{SELECT}_{\widetilde{L}})
=
\sum_{k=1}^{n}\mathcal{O}(k)
=
\mathcal{O}(n^2),
\end{equation}
and
\begin{equation}
T_{\mathrm{cmp}}^{\mathrm{CNOT}}(\mathrm{SELECT}_{\widetilde{L}})
=
\sum_{k=1}^{n}\mathcal{O}(k^2)
=
\mathcal{O}(n^3).
\end{equation}
Hence the appendix compilation improves the direct recursive realization from $\mathcal{O}(n^3)$ to $\mathcal{O}(n^2)$ in macro block count, and from $\mathcal{O}(n^4)$ to $\mathcal{O}(n^3)$ in primitive CNOT count.
\end{remark}

We evaluate the performance of the constructions introduced in Secs.~\ref{sec:quantum_representation_of_convolution} and \ref{sec:circuits} against classical FFT-based convolution and prior quantum formulations~\cite{Castelazo2022}. 
Throughout this section, we assume $N=2^n$.

\myheading{Input models.}
We distinguish two practically relevant input settings throughout this section.
In the coherent quantum-state input model, both $\ket{\bm{a}}$ and $\ket{\bm{b}}$ are already available as quantum states, for example from an upstream quantum routine or direct hardware state preparation, so no loading cost is charged to the convolution subroutine itself.
In the classical-vector input model, both $\bm{a}$ and $\bm{b}$ are classical data and must be loaded/prepared, and we denote the total overhead of this step by $T_{\mathrm{load}}(N)$.
In this coherent-input setting, the circuit consumes the supplied states directly. Because the LCU construction is asymmetric, no additional kernel-dependent $\mathrm{PREP}_b$ oracle needs to be charged inside the convolution block itself.
This separation matches the main use case emphasized in this work: the circuit is most effective in the coherent quantum-state input model, whereas in the classical-vector input model the dominant additional cost is typically state preparation rather than the convolution block itself.

\myheading{Complexity-model conventions.}
Unless noted otherwise, Table~\ref{tab:complexity_comparison} reports arithmetic complexity for classical algorithms. For the modular adder LCU construction introduced here, the table intentionally suppresses the distinction among direct recursive, optimized bitwise recursive, QFT adder, and ripple-carry adder realizations and records only their shared polylogarithmic scaling in $N$; the more specific realization-dependent costs are stated explicitly later in this section. Furthermore, precision-dependent fault-tolerant synthesis overhead for arbitrary-angle rotations (e.g., Clifford+$T$ decomposition of QFT phases) is not included in the primitive-gate asymptotics listed in the table.

We distinguish three notions of cost in what follows: (i) deterministic per-shot circuit size of the convolution block; (ii) expected runtime after postselection, or after postselection plus amplitude amplification, which introduces an additional factor depending on $p_{\mathrm{succ}}$; and (iii) end-to-end runtime in the classical-vector input model, which additionally includes the loading cost $T_{\mathrm{load}}(N)$. Table~\ref{tab:complexity_comparison} records per-shot asymptotic circuit scaling, while the postselection and loading overheads are stated separately in the subsequent discussion and in the note below the table.

\myheading{Recursive construction and optimized compilation.}
Classical direct convolution uses $\mathcal{O}(N^2)$ arithmetic operations, while classical FFT-based convolution uses $\mathcal{O}(N\log N)$~\cite{Knuth1998,Cooley1965}. 
For the recursive reflected-shift construction, one application of the $\mathrm{SELECT}_{\widetilde{L}}$ block has direct structural cost $\mathcal{O}(\log^3 N)$ in macro blocks and $\mathcal{O}(\log^4 N)$ in primitive CNOT. 
However, the optimized bitwise recursive compilation (Appendix~\ref{app:compiled_suffix_recursion} and Proposition~\ref{prop:exact_equiv}) successfully reduces this cost to $\mathcal{O}(\log^2 N)$ in macro blocks and $\mathcal{O}(\log^3 N)$ in primitive CNOT count under standard decomposition assumptions~\cite{Claudon2024,Rosa2025,Zindorf2025}. 

Hence, in the coherent quantum-state input model, this optimized realization yields a per-shot runtime of $\mathcal{O}(\log^2 N)$ (macro block) or $\mathcal{O}(\log^3 N)$ (primitive CNOT). 
In the classical-vector input model, where data loading is required, the total runtime naturally becomes $\mathcal{O}(T_{\mathrm{load}}(N)+\log^2 N)$ or $\mathcal{O}(T_{\mathrm{load}}(N)+\log^3 N)$, respectively.

Since the LCU step is probabilistic, we must account for the success probability $p_{\mathrm{succ}}$. Using Remark~\ref{rem:ab_normalization}, this is given by
\begin{equation}
p_{\mathrm{succ}}
=\frac{\|C_n(\bm{b})\bm{a}\|_2^2}
{N\,\|\bm{a}\|_2^2\,\|\bm{b}\|_2^2}.
\end{equation}
Using Eq.~\eqref{eq:C_from_H}, this is equivalent to
\begin{equation}
p_{\mathrm{succ}}
=\frac{\|H_n(\bm{b})\bm{a}_{\mathrm{R}}\|_2^2}
{N\,\|\bm{a}\|_2^2\,\|\bm{b}\|_2^2},
\qquad \bm{a}_{\mathrm{R}}:=J_n\bm{a}.
\end{equation}
For the optimized compiled realization, amplitude amplification gives the expected runtime
\begin{equation}
\mathcal{O}\!\left(
\frac{T_{\mathrm{load}}(N)+\log^2 N}{\sqrt{p_{\mathrm{succ}}}}
\right)
\end{equation}
in the compiled macro block model, which reduces in the coherent-input setting to
\begin{equation}
\mathcal{O}\!\left(
\frac{\log^2 N}{\sqrt{p_{\mathrm{succ}}}}
\right).
\end{equation}
Under the primitive-CNOT model, the same compiled realization gives
\begin{equation}
\mathcal{O}\!\left(
\frac{T_{\mathrm{load}}(N)+\log^3 N}{\sqrt{p_{\mathrm{succ}}}}
\right)
\end{equation}
and, in the coherent-input setting,
\begin{equation}
\mathcal{O}\!\left(
\frac{\log^3 N}{\sqrt{p_{\mathrm{succ}}}}
\right).
\end{equation}
For alternative adder realizations of Theorem~\ref{thm:asym_modadder_conv}, the same expressions hold with the $\log^2 N$ or $\log^3 N$ term replaced by the corresponding realization cost.
In particular, the ripple-carry realization of Sec.~\ref{subsec:qft_equiv} gives $\mathcal{O}(\log N)$ primitive-gate adder complexity, while the exact QFT adder gives $\mathcal{O}(\log^2 N)$.
Thus all modular adder LCU realizations discussed in this work remain polylogarithmic in $N$, even though the precise exponent and cost model depend on the synthesis route.
As usual in amplitude-encoded algorithms, these quantum runtimes refer to preparing the output state or estimating observables from it; reconstructing the full classical output vector would require separate measurement/tomography overhead.

\myheading{Space / memory footprint.}
Classical convolution algorithms require storing $N$ complex numbers, corresponding to $\mathcal{O}(N)$ memory. 
When the inputs are already coherent quantum states, the circuit itself requires $\mathcal{O}(\log N)$ qubits overall, with the exact ancilla count depending on the chosen adder synthesis and on the decomposition of multi-controlled operations. This is the regime where amplitude encoding gives exponential compression of the in-circuit state representation. If the inputs start as classical vectors, one must additionally account for classical storage/loading resources; in this section those costs are summarized through the $T_{\mathrm{load}}(N)$ notation rather than repeated row-by-row in the table.

Table~\ref{tab:complexity_comparison} therefore reports the same methods in the two input models most relevant here:
(i) coherent quantum-state inputs, and
(ii) classical-vector inputs requiring state preparation/loading.
For the Fourier-domain approach of Ref.~\cite{Castelazo2022}, however, one must additionally assume oracle access to Fourier-domain kernel entries, or else charge the cost of constructing that oracle; we denote any such extra cost by $T_{\mathrm{F\mbox{-}oracle}}(N)$. We list the Fourier-domain approach separately because it is a distinct algorithmic mechanism that implements convolution in Fourier space. By contrast, the direct recursive construction, the optimized bitwise recursive compilation, and the QFT adder and ripple-carry adder syntheses are different realizations of the same LCU modular adder convolution block.

\begin{table*}[tbp]
\centering
\scriptsize
\caption{Summary asymptotic scaling for convolution on $N=2^n$ data points under the coherent-state and classical-vector input models. The classical rows report arithmetic runtime. The quantum rows report per-shot block cost in the stated oracle/input models. For the Fourier-domain row, any nontrivial cost of providing Fourier-domain kernel access is isolated into $T_{\mathrm{F\mbox{-}oracle}}(N)$; success-probability overhead and further caveats are stated separately in the text and note.}
\label{tab:complexity_comparison}
\begin{tabularx}{\textwidth}{@{} >{\raggedright\arraybackslash\hsize=1.1\hsize}X >{\centering\arraybackslash\hsize=0.95\hsize}X >{\centering\arraybackslash\hsize=1.15\hsize}X >{\centering\arraybackslash\hsize=0.8\hsize}X @{}}
\toprule
\textbf{Method}
& \textbf{Coherent Quantum-State Inputs}
& \textbf{Classical-Vector Inputs}
& \textbf{Circuit / Memory Footprint} \\
\midrule

Classical direct 
& --- 
& $\mathcal{O}(N^2)$ 
& $\mathcal{O}(N)$ \\
\addlinespace

Classical FFT 
& --- 
& $\mathcal{O}(N \log N)$ 
& $\mathcal{O}(N)$ \\
\addlinespace

Fourier-domain approach~\cite{Castelazo2022}
& $\mathcal{O}(T_{\mathrm{F\mbox{-}oracle}}(N)+\log^2 N)$
& $\mathcal{O}(T_{\mathrm{load}}(N)+T_{\mathrm{F\mbox{-}oracle}}(N)+\log^2 N)$
& $\mathcal{O}(\log N)$ qubits \\
\addlinespace

\textbf{LCU modular adder} (direct recursive, optimized bitwise recursive, QFT, ripple-carry)
& $\mathcal{O}(\operatorname{polylog}(N))$ 
& $\mathcal{O}\left(T_{\mathrm{load}}(N)+\operatorname{polylog}(N)\right)$
& $\mathcal{O}(\log N)$ qubits \\
\addlinespace

\bottomrule
\end{tabularx}
\begin{flushleft}
\scriptsize
\textit{Note:} \\
$\bullet$ $T_{\mathrm{load}}(N)$ denotes state-preparation/loading overhead and is absent when the relevant data input is already available as a coherent quantum state. \\
$\bullet$ $T_{\mathrm{F\mbox{-}oracle}}(N)$ denotes any additional preprocessing / oracle-construction cost required to provide Fourier-domain access to the kernel, as assumed in the Fourier-space block-encoding framework of Ref.~\cite{Castelazo2022}; when such oracle access is given analytically or by assumption, this term may be omitted. \\
$\bullet$ The table reports per-shot block costs before success-probability amplification overhead. 
For the modular adder LCU constructions, postselected runtimes acquire an additional factor $1/\sqrt{p_{\mathrm{succ}}}$. For the Fourier-domain approach, Ref.~\cite{Castelazo2022} likewise analyzes end-to-end output-state preparation via block-encoding and amplitude amplification; the table records only the underlying block cost. \\
$\bullet$ The unified LCU row suppresses realization-dependent details: the direct recursive realization has $\mathcal{O}(\log^3 N)$ macro block complexity and $\mathcal{O}(\log^4 N)$ primitive CNOT complexity, Appendix~\ref{app:compiled_suffix_recursion} provides an exactly equivalent optimized bitwise recursive compilation with $\mathcal{O}(\log^2 N)$ macro block complexity and $\mathcal{O}(\log^3 N)$ primitive CNOT complexity, and exact QFT adder and ripple-carry realizations have $\mathcal{O}(\log^2 N)$ and $\mathcal{O}(\log N)$ primitive logical-gate complexity, respectively.
\end{flushleft}
\end{table*}

\section{Block-Encoding and QSVT Implementation}
\label{sec:blockencoding_qsvt}
In this section, we show that the convolution operator $H_n(\bm{b})$ admits a direct block-encoding from the same asymmetric LCU pipeline used for circuit synthesis and is therefore naturally compatible with QSVT. For real-valued kernels, $H_n(\bm{b})$ is strictly Hermitian. Thus the $J_n$-symmetrized construction provides a native Hermitian representative of circular convolution on the original data register, up to the explicit relation $C_n(\bm{b})=H_n(\bm{b})J_n$. We use deconvolution as a concrete example. We do not claim an asymptotic advantage over direct non-Hermitian singular-value pseudoinversion or over Hermitian-dilation-based singular-value processing; rather, the main point is that the present reformulation is intrinsically Hermitian, and relative to the auxiliary normal-equation route it avoids condition-number squaring.

\subsection{Existence of a Block-Encoding}

\begin{corollary}[Direct asymmetric block-encoding of $H_n(\bm{b})$]
Let $N=2^n$ and assume the state-preparation unitary $\mathrm{PREP}_b$ acts on the $n$-qubit ancilla register $A$ to
\begin{equation}
\mathrm{PREP}_b\ket{0}^{\otimes n}
=\ket{\bm{b}}
=\sum_{i=0}^{N-1} b_i\ket{i},
\qquad
\sum_i |b_i|^2=1.
\end{equation}
Define the overall unitary $\mathcal{U}_H$ as
\begin{equation}
\mathcal U_H
:=
(\mathrm{PREP}_u^\dagger\otimes I)\,
\mathrm{SELECT}_{\widetilde{L}}\,
(\mathrm{PREP}_b\otimes I),
\end{equation}
where $\mathrm{PREP}_u\ket{0}^{\otimes n}=\ket{u}:=\frac{1}{\sqrt N}\sum_{i=0}^{N-1}\ket{i}$ and
\begin{equation}
\mathrm{SELECT}_{\widetilde{L}}
:=\sum_{i=0}^{N-1}\ket{i}\!\bra{i}\otimes \widetilde{L}_{i,n}.
\end{equation}
Then, projecting the ancilla register back to $\ket{0}_A^{\otimes n}$ yields
\begin{equation}
(\bra{0}^{\otimes n}\otimes I)\,\mathcal{U}_H\, 
(\ket{0}^{\otimes n}\otimes I)
=\frac{1}{\sqrt N}\,H_n(\bm{b}),
\end{equation}
where $H_n(\bm{b})=\sum_{i=0}^{N-1} b_i\widetilde{L}_{i,n}$.
\end{corollary}

\begin{proof}
Using $(\bra{0}^{\otimes n}\mathrm{PREP}_u^\dagger)=\bra{u}$ and
$\ket{\bm{b}}=\sum_i b_i\ket{i}$, we have
\begin{equation}
\begin{split}
(\bra{0}^{\otimes n}_A\otimes I_D)\mathcal U_H(\ket{0}_A^{\otimes n}\otimes I_D)
&=(\bra{u}_A\otimes I_D)\,\mathrm{SELECT}_{\widetilde{L}}\,(\ket{\bm{b}}_A\otimes I_D) \\
&=\frac{1}{\sqrt N}\sum_{i=0}^{N-1} b_i\,\widetilde{L}_{i,n}
=\frac{1}{\sqrt N}\,H_n(\bm{b}).
    \end{split}
\end{equation}
\end{proof}

Consequently, $\mathcal U_H$ is an $(\alpha,n,0)$-block-encoding of $H_n(\bm b)$ with $\alpha=\sqrt N\,\|\bm b\|_2$, which reduces to $\alpha=\sqrt N$ for normalized kernels.
For an input state $\ket{\psi}$, the corresponding success probability of the postselection measurement is
\begin{equation}
p_{\mathrm{succ}} = \frac{\|H_n(\bm{b})\ket{\psi}\|_2^2}
{N\,\|\bm{b}\|_2^2}.
\end{equation}

\subsection{Hermitian Structure for Real Kernels}

A key structural feature of our proposed construction is the symmetrizing role of the reversal operator $J_n$.

\begin{lemma}[Hermiticity for real kernels]\label{lemma:hermiticity}
If $\bm{b} \in \mathbb{R}^{2^n}$, then
\begin{equation}
H_n(\bm{b})^\dagger = H_n(\bm{b}).
\end{equation}
\end{lemma}

\begin{proof}
Using $J_n=J_n^\dagger = J_n^{-1}$ and $L_{i,n}^\dagger = L_{-i,n}$.
A fundamental property of the reversal operator $J_n$ is its relation with the cyclic operators: $J_n L_{-i,n} J_n = L_{i,n}$, which implies that conjugating a shift by reversal operator inverts its direction.
Using its relation, we can compute the adjoint of the reflected shift operator $\widetilde{L}_{i,n}$ as follows: 
\begin{equation}
\widetilde{L}_{i,n}^\dagger
=(L_{i,n} J_n )^\dagger
=J_n L_{-i,n}
=L_{i,n} J_n
=\widetilde{L}_{i,n}.
\end{equation}
This confirms each reflected shift operator $\widetilde{L}_{i,n}$ is Hermitian.
For a real-valued kernel $\bm{b} \in \mathbb{R}^{2^n}$, each element satisfies $b_i^* =b_i$.
Consequently, the convolution operator $H_n(\bm{b})$ is also Hermitian.
\begin{equation}
H_n(\bm{b})^\dagger
=
\sum_i b_i^* \widetilde{L}_{i,n}^\dagger
=
\sum_i b_i \widetilde{L}_{i,n}=H_n(\bm{b}).
\end{equation}
\end{proof}

This Hermitian symmetry is not inherent to generic shift-based constructions and arises directly from the symmetrization by $J_n$.

\subsection{Hermitian Block-Encoding and QSVT for Deconvolution}

Given an $(\alpha,n,0)$-block-encoding of $H_n(\bm b)$, the QSVT framework enables the application of degree-$d$ polynomial transformations to the operator with a query complexity of $\mathcal{O}(d)$~\cite{Low2017,Gilyen2019,Low2019}. 
We focus on the deconvolution problem considered in~\cite{Castelazo2022} as a concrete example, which requires implementing (or approximating) the inverse of a convolution operator, in the broader spirit of quantum linear-systems and related inverse-type quantum algorithms~\cite{HHL2009,Childs2017,Chakraborty2019}.

\myheading{Direct non-Hermitian route and auxiliary Hermitianizations.} 
In the standard LCU framework for group convolution~\cite{Castelazo2022}, the convolution operator
\begin{equation}
C_n(\bm{b}) = \sum_i b_i L_{i,n}
\end{equation}
remains generally non-Hermitian, even for real-valued kernels.
Hereafter we abbreviate $C_n(\bm{b})$ as $C_n$.
Because $C_n$ is non-Hermitian, the natural block-encoding/QSVT framework acts on its singular values rather than directly implementing an arbitrary matrix function $f(C_n)$ on an eigenvalue spectrum~\cite{Gilyen2019}. In particular, singular-value processing of a non-Hermitian operator is not the same object as, for example, the matrix exponential $\mathrm{e}^{C_n}$.
For inverse / pseudoinverse-type tasks, however, direct non-Hermitian singular-value pseudoinversion is in principle available and can already achieve linear dependence on the effective condition number~\cite{Gilyen2019,Castelazo2022}.

The same singular-value viewpoint may be expressed through the standard Hermitian dilation
\begin{equation}
\mathcal{H}(C_n) =
\begin{pmatrix}
0 & C_n \\
C_n^\dagger & 0
\end{pmatrix}.
\end{equation}
This doubled-register Hermitian embedding has eigenvalues $\pm \sigma_i(C_n)$, so it encodes the same singular values as $C_n$. Accordingly, inverse / pseudoinverse-type singular-value processing via $\mathcal{H}(C_n)$ also retains linear dependence on $\kappa(C_n)$~\cite{Gilyen2019}.

A different auxiliary Hermitianization is the positive-semidefinite normal-equation operator $C_n^\dagger C_n$, for which the relevant spectral domain shifts from singular values $\sigma_i$ to eigenvalues $\lambda_i=\sigma_i^2$. Consequently, the condition number is squared:
\begin{equation}
\kappa(C_n^\dagger C_n)=\frac{\sigma_{\max}^2}{\sigma_{\min}^2}=\kappa(C_n)^2.
\end{equation}
Thus, if one chooses the normal-equation route as the comparison baseline, the required polynomial degree for inverse-type processing becomes
\begin{equation}
d(C_n^\dagger C_n)=\mathcal{O}\!\left(\kappa(C_n)^2\log(1/\varepsilon)\right),
\label{eq:qsvt_degree_normal_eq}
\end{equation}
where $\varepsilon$ denotes the allowable approximation error~\cite{Gilyen2019}. 
We emphasize that this quadratic $\kappa(C_n)^2$ dependence is specific to the auxiliary normal-equation route. It is not a fundamental limitation of direct non-Hermitian singular-value transformation, nor of Hermitian-dilation-based singular-value processing.

\myheading{Hermitian formulation via the reversal matrix.}
In contrast to the standard framework, our symmetrized operator $H_n(\bm{b})$ from Eq.~\eqref{eq:convolution_operator}, built from the reflected shifts $\widetilde{L}_{i,n}=L_{i,n}J_n$, is strictly Hermitian for any real-valued kernel $\bm{b}$ (Lemma~\ref{lemma:hermiticity}). 
QSVT therefore acts directly on its eigenvalues on the original $n$-qubit data space, without introducing a doubled dilation register and without passing through normal equations.

Leveraging this Hermitian property for deconvolution, the QSVT polynomial degree $d$ required to approximate $H_n(\bm{b})^{-1}$ scales as
\begin{equation}
d(H_n(\bm{b})) 
= \mathcal{O}\!\left(\kappa(H_n(\bm{b})) \log(1/\varepsilon)\right).
\label{eq:qsvt_degree_H}
\end{equation}
Relative to the normal-equation degree estimate in Eq.~\eqref{eq:qsvt_degree_normal_eq}, this native Hermitian formulation avoids condition-number squaring and therefore lowers the required polynomial degree to the linear-in-$\kappa$ scaling stated in Eq.~\eqref{eq:qsvt_degree_H}. By contrast, its $\kappa$-dependence matches the direct non-Hermitian singular-value route and the standard Hermitian dilation route for inverse / pseudoinverse-type processing. The point is therefore not a universal asymptotic improvement over all non-Hermitian methods, but rather that the $J_n$-symmetrized operator supplies a native Hermitian representative of the convolution family itself.

\begin{remark}[Assumptions behind inverse-QSVT degree estimates]
\label{rem:qsvt_inverse_assumptions}
The degree estimates in Eqs.~\eqref{eq:qsvt_degree_normal_eq} and \eqref{eq:qsvt_degree_H} rely on standard inversion promises:
\begin{enumerate}
\item The target operator is rescaled into a block-encoding domain (spectral norm at most $1$);
\item Inversion is required only on a promised nonzero spectral interval, e.g.
$[-1,-1/\kappa]\cup[1/\kappa,1]$ for the Hermitian operator $H_n(\bm{b})$, or equivalently on singular values in $[1/\kappa,1]$ for the non-Hermitian route;
\item Equivalently, if exact invertibility is absent, one uses a thresholded pseudoinverse with effective condition number set by the threshold.
\end{enumerate}
Under these assumptions, standard polynomial-approximation arguments for the inverse function $1/x$ on the promised interval yield the quoted $\mathcal{O}(\kappa\log(1/\varepsilon))$-type dependence for the direct singular-value route, for Hermitian dilation, and for the present native Hermitian formulation~\cite{Gilyen2019}. By contrast, the normal-equation comparison route incurs the squared $\mathcal{O}(\kappa^2\log(1/\varepsilon))$ dependence because its spectrum is $\sigma_i^2$.
\end{remark}

\myheading{Complexity comparison.}
Since $C_n$ and $H_n(\bm{b})$ share the same singular value spectrum (due to the unitarity of $J_n$), their condition numbers are identical:
\begin{equation}
\kappa(H_n(\bm{b})) = \kappa(C_n).
\end{equation}
\begin{proposition}[Deconvolution mapping between $C_n$ and $H_n$]
\label{prop:deconv_mapping}
From Eq.~\eqref{eq:C_from_H}, we have
\begin{equation}
C_n(\bm{b})=H_n(\bm{b})J_n.
\end{equation}
If $C_n(\bm{b})$ is invertible, then
\begin{equation}
C_n(\bm{b})^{-1}=J_nH_n(\bm{b})^{-1},
\end{equation}
Hence solving deconvolution for $C_n(\bm{b})$ is equivalent to solving it for the symmetrized operator $H_n(\bm{b})$ followed by a straightforward $J_n$ transformation.
\end{proposition}

\begin{proof}
Right-multiplying Eq.~\eqref{eq:C_from_H} by $J_n$ yields $C_n(\bm{b}) = H_n(\bm{b}) J_n$. 
Taking the inverse of both sides gives $C_n(\bm{b})^{-1} = (H_n(\bm{b}) J_n)^{-1} = J_n^{-1} H_n(\bm{b})^{-1} = J_n H_n(\bm{b})^{-1}$.
\end{proof}

For the normal-equation degree estimate in Eq.~\eqref{eq:qsvt_degree_normal_eq} and the native-Hermitian degree estimate in Eq.~\eqref{eq:qsvt_degree_H}, the polynomial-degree ratio is
\begin{equation}
\frac{d(C_n^\dagger C_n)}{d(H_n(\bm{b}))}
=
\Theta(\kappa(H_n(\bm{b}))).
\end{equation}
Thus the benefit of the present Hermitian reformulation is best understood as avoiding the extra factor of $\kappa$ introduced by this particular normal-equation comparison route, while also keeping the deconvolution map $C_n(\bm{b})^{-1}=J_nH_n(\bm{b})^{-1}$ explicit on the original data register.

\section{Discussion}
\myheading{Unifying view from Theorem~\ref{thm:asym_modadder_conv}.}
Theorem~\ref{thm:asym_modadder_conv} identifies a common core: modular addition on computational-basis indices is formally equivalent to realizing the asymmetric-LCU SELECT operation required for circular convolution. 
This immediately unifies circulant-matrix implementations~\cite{Wang2017}, related structured-matrix approaches such as Toeplitz systems~\cite{Wan2018}, and group-convolution LCU framework~\cite{Castelazo2022}.

\myheading{Adder choices and realization trade-offs.}
The direct recursive construction of the main text, its optimized bitwise compilation (Appendix~\ref{app:bitwise_compilation} and Proposition~\ref{prop:exact_equiv}), and standard QFT or ripple-carry adders all produce equivalent middle blocks for the target $\mathrm{SELECT}_{\widetilde{L}}$ (with QFT and ripple-carry routes mapped through Eq.~\eqref{eq:select_bridge}). 
Therefore, they should be viewed as interchangeable implementations of the same target unitary.
Among the architectures studied here, direct ripple-carry scales as $\mathcal{O}(n)$ in primitive logical gates, while the exact QFT adder scales as $\mathcal{O}(n^2)$.
For our recursive methods, the direct structural realization scales as $\mathcal{O}(n^3)$ in macro blocks and $\mathcal{O}(n^4)$ in primitive CNOT count, whereas the optimized bitwise compilation successfully reduces these costs to $\mathcal{O}(n^2)$ in macro blocks and $\mathcal{O}(n^3)$ in primitive CNOT count. 
The primary role of the direct recursion is therefore explanatory rather than gate-optimal: it explicitly exposes the shift algebra and the $J_n$ symmetrization, while the bitwise compilation provides a more efficient exact compilation of the same operator. 
Furthermore, modern optimized modular addition directions can also be used within the same theorem-level equivalence, including recent lookahead and carry-save variants~\cite{Gaur2024LogDepthAdder,Mitra2025QCSA,Remaud2025AncillaFree,Wang2025}.

\myheading{Contribution scope and complexity optimality.}
The direct recursive construction is not claimed to asymptotically outperform the best standalone modular adders. 
We also do not claim novelty for the identity $\mathrm{ADD}_N=\mathrm{SELECT}_L$ itself, nor for individual ripple-carry and QFT adder primitives. 
Rather, our contributions are threefold: (i) an explicit asymmetric-LCU formulation bridging modular addition and circular convolution, which inherently preserves complex amplitudes while rigorously accounting for normalization and success probabilities; 
(ii) the $J_n$-symmetrized operator pipeline and its associated Hermitian deconvolution mapping, and (iii) a recursion-first structural construction together with an exactly equivalent optimized bitwise compilation and a proof of its block-by-block equivalence with standard adder implementations.
Accordingly, the advantage claimed here is not reduced arithmetic complexity, nor a universal asymptotic improvement over all non-Hermitian inverse implementations.
Instead, the proposed framework offers improved theoretical analyzability, modular circuit compilation, and a direct Hermitian route for inverse-type spectral transforms via the $J_n$-symmetrized pipeline.

\myheading{Input model and practical efficiency.}
The proposed formulation is most relevant in practice in the regime where both the data and the kernel are already available as coherent quantum resources, such as from an upstream quantum routine or direct hardware state preparation. 
In this setting, the asymmetric construction offers a distinct advantage: the postselection state is fixed at $\ket{u}$, so the only uncompute operation inside the convolution block is $\mathrm{PREP}_u^\dagger$. 
The kernel enters solely through the supplied input ancilla state $\ket{\bm b}$, eliminating the need for a kernel-dependent inverse preparation $\mathrm{PREP}_b^\dagger$ during convolution. 
Furthermore, this asymmetry naturally preserves the complex coefficients $b_i$; a symmetric overlap would instead produce $|b_i|^2$ weights and erase their phases. 
This state-input paradigm should be distinguished from the oracle-access model of Ref.~\cite{Castelazo2022}.
In their model, while the data may also be coherent, the kernel components (either in the spatial or Fourier domain) must be queried individually as classical bit-strings via an oracle, rather than being supplied directly and coherently as the amplitude-encoded input state $\ket{\bm b}$. 

A representative use case is coherent deblurring or inverse filtering inside a larger quantum signal-processing pipeline: $\ket{\bm a}$ is produced by an upstream quantum simulation or quantum sensing subroutine, while $\ket{\bm b}$ encodes a translation-invariant point-spread or transfer kernel supplied coherently by an upstream preparation routine. 
In such a setting, the convolution block itself contributes only the modular adder cost, the asymmetric construction avoids charging a kernel-dependent $\mathrm{PREP}_b^\dagger$ inside the block, and the $J_n$-symmetrized formulation provides a direct Hermitian route for subsequent inverse-type spectral transforms such as deconvolution. 
By contrast, for fully classical input, the dominant practical bottleneck is typically the state preparation and loading cost $T_{\mathrm{load}}(N)$, rather than the convolution block.

\myheading{Role of $J_n$ and method choice.}
The role of $J_n$ is structural rather than cosmetic: it trades one inexpensive bitwise NOT layer for a formulation that exposes an explicit recursion, admits the Pauli-support characterization used in Appendix~\ref{app:pauli_support}, and yields a Hermitian operator for real-valued kernels. 
Accordingly, optimized modular adders are the natural choice when arithmetic efficiency is the primary goal, whereas the recursive-$J_n$ formulation is preferable when structural transparency and direct Hermitian spectral processing are the main priorities.

\section{Conclusion}
We have reformulated quantum circular convolution around a single circuit-level statement: discrete circular convolution is realized exactly as an asymmetric-LCU block when the controlled-shift unitary is implemented by modular addition (Theorem~\ref{thm:asym_modadder_conv}). This establishes a direct bridge between the circulant-adder-based circulant matrix implementations~\cite{Wang2017} and the broader group-convolution LCU perspective~\cite{Castelazo2022}.

Building upon this equivalence, we introduced the reversal matrix $J_n$ to construct $\widetilde{L}_{i,n}=L_{i,n}J_n$ and the corresponding convolution operator $H_n(\bm{b})=\sum_i b_i\widetilde{L}_{i,n}$. 
This symmetrized form yields explicit recursive structure, an optional Pauli-string expansion path, and, crucially, Hermiticity for real-valued kernels; for inverse-type spectral transforms, it therefore provides a direct Hermitian route, which we explicitly compared against the auxiliary normal-equation baseline in Sec.~\ref{sec:blockencoding_qsvt}.

We also provided a structurally transparent recursive construction paired with an exactly equivalent optimized bitwise compilation, and compared these implementations against standard QFT and ripple-carry adders under explicit cost-model conventions. 
The resulting framework gives a unified path from operator formulation to circuit synthesis.
It is particularly effective when the relevant operands are already available as coherent quantum states; for classical vectors, the dominant bottleneck typically lies in state preparation and loading rather than in the convolution circuit itself.

\section{Acknowledgments}

Hiroshi C. Watanabe, Norio Yoshida and Sergey Gusarov were supported by CSTIP grant from National Research Council Canada (AQC203).
Hiroshi C. Watanabe was supported by Council for Science, Technology and Innovation (CSTI), Cross-ministerial Strategic Innovation Promotion Program (SIP), “Promoting the application of advanced quantum technology platforms to social issues” (Funding agency:QST), JSPS KAKENHI Grant Numbers 23K03266, and the Quantum and Spacetime Research Institute (QuaSR), Kyushu University.
This work used the computational resources of the Research Center for Computational Science, Okazaki (Project: 25-IMS-C159 and 25-IMS-C068), the MCRP-S at the Center for Computational Sciences, University of Tsukuba, and the Program for Promoting Researches on the Supercomputer Fugaku (“Data-Driven Research Methods Development and Materials Innovation Led by Computational Materials Science,” (JPMXP1020230327) and Project ID: hp250229. 
Norio Yoshida was also supported by JSPS KAKENHI Grant Numbers 24K01434, 24H00282, and 25H00913, and by the Toyota Physical and Chemical Research Institute (Toyota RIKEN) scholar collaboration research Phase-II.

\appendix

\section{Structural Recursion of the Reflected Generator}
\label{app:Un_recursion}
In this section we derive a structural recursion for the reflected generator $U_n=\widetilde{L}_{1,n}=L_{1,n}J_n$. The purpose of this recursion is not to provide the most gate-efficient arithmetic implementation of $L_{1,n}$ itself, but to expose the internal operator structure of the $J_n$-symmetrized reflected-shift family. The optimized compiled realization discussed in Sec.~\ref{sec:complexity} is instead based on the standard binary carry-propagation recursion of the increment operator.

\myheading{Definition of the reflected generator.}
We define the reflected generator $U_n$ as
\begin{equation}
    U_n :=  L_{1,n} J_n,
\end{equation}
where $L_{1,n}$ denotes a cyclic shift by one element in the group.
For $n=1$, the left-regular representation of $\mathbb{Z}/2\mathbb{Z}$ is
\begin{equation}
    L_{1,1} =
    \begin{bmatrix}
        0 & 1 \\
        1 & 0
    \end{bmatrix}
    = X,
\end{equation}
leading to $U_1 = L_{1,1} J_1 = X X = I$.
Based on the actions of $L_{1,n}$ and $J_n$ defined in Eqs.~\eqref{eq:L_action_on_ket} and \eqref{eq:J_action_on_ket}, the action of $U_n$ on a computational basis state $\ket{k}$ is given by
\begin{equation}\label{eq:U_action_on_ket}
    U_n\ket{k}=\ket{2^n - k \pmod{2^n}}
\end{equation}

\myheading{Block structure and projectors.}
\label{sec:recursive_un}
To facilitate the recursive derivation, 
we group the $n$ qubits into $n-1$ ``higher-order'' qubits and one ``least significant'' qubit.  
A computational basis state is expressed as
\begin{equation}
    \ket{x,b}, \quad x = x_{n-1} \cdots x_1 \in \{0,1\}^{n-1},\; b\in\{0,1\},
\end{equation}
where $\ket{x,b} = \ket{x_{n-1} \cdots x_1 b}$ and $b$ denotes the LSB in the binary index.

In this ordered basis, any operator acting on $n$ qubits can be represented as a $2\times 2$ block matrix, where each block acts on the $n-1$ higher-order qubits.
The projectors onto the LSB are
\begin{equation}
    \ket{0}\!\bra{0} = \frac{I+Z}{2}, \qquad
    \ket{1}\!\bra{1} = \frac{I-Z}{2}.
\label{eq:projectors}
\end{equation}
Consequently, a block-diagonal operator $V$ can be represented as
\begin{equation}
   V =  A \otimes \ket{0}\!\bra{0} + B \otimes \ket{1}\!\bra{1}
\end{equation}
where $A,B$ are operators acting on the first $n-1$ qubits.

\begin{lemma}[Recursive form of $U_n$]
\label{lem:Un-recursion}
The generator $U_n$ satisfies the following recursive relation,
\begin{equation}
    U_{n+1}
    =U_{n} \otimes \ket{0}\!\bra{0}
    +J_n \otimes \ket{1}\!\bra{1}
\end{equation}
with the base case $U_1=I$.
Furthermore, $U_n$ can be expressed in the following expanded form:
\begin{equation}
U_n =
I \otimes \ket{0}\!\bra{0}^{\otimes (n-1)}
+
\sum_{i=1}^{n-1}
\left(
J_i \otimes \ket{1}\!\bra{1} \otimes \ket{0}\!\bra{0}^{\otimes (n-i-1)}
\right).
\end{equation}
where $J_i:=X^{\otimes i}$.
\end{lemma}

\begin{proof}
We verify the recursive relation by considering the base case and then deriving the block structure of $U_{n+1}$ for the general case.

{\bf Base case: ($n=1$)}: 

\noindent From the definition $U_n = L_{1,n}J_n$, we have $U_1 = L_{1,1}J_1 = X X = I$.
Applying the recursive formula \eqref{eq:Un-recursion} for $n=1$ yields:\begin{equation}U_2 = U_1 \otimes \ket{0}\!\bra{0} + J_1 \otimes \ket{1}\!\bra{1} = I \otimes \ket{0}\!\bra{0} + X \otimes \ket{1}\!\bra{1}.
\end{equation}
Direct calculation of $U_2 = L_{1,2}J_2$ confirms that its action on the basis states matches this block-diagonal form (e.g., $U_2\ket{00} = \ket{00}$ and $U_2\ket{01} = \ket{11}$).

{\bf General case: ($n\ge1$)}: 

\noindent We derive the recursive block form of $U_{n+1}=L_{1,n+1}J_{n+1}$ by evaluating its action on the ordered basis $\ket{k,b}$, where $k \in \{0,\cdots,2^n-1\}$ and $b\in\{0,1\}$.
From the definition of the operator, $J_{n+1}\ket{k, b} = \ket{2^n-1-k, 1-b}$.

\begin{equation}
\begin{split}
     U_{n+1}\ket{k,b}
 &=L_{1,{n+1}}~J_{n+1}\ket{k,b}\\
 &=L_{1,{n+1}}~\left( X^{\otimes n}\ket{k}\otimes X\ket{b}\right) \\
 &=L_{1,{n+1}} \left(\ket{2^{n} -1-k} \ket{b\oplus1} \right)
\end{split}
\end{equation}

\paragraph{Case 1: $b=0$}
\begin{equation}
\begin{split}
U_{n+1}\ket{k, 0} &= L_{1,n+1} (\ket{2^n-1-k} \otimes \ket{1}) \\
&= L_{1,n+1} \ket{2(2^n-1-k) + 1} \\
&= \ket{2(2^n-1-k) + 2 \pmod{2^{n+1}}} \\
&= \ket{2(2^n-k)} \\
&= \ket{2^n-k \pmod{2^n}} \otimes \ket{0}\\
&=U_{n}\ket{k}\otimes\ket{0}
\end{split}
\end{equation}
where the last equality is based on Eq~\eqref{eq:U_action_on_ket}.

\paragraph{Case 2: $b=1$}
\begin{equation}
\begin{split}
U_{n+1}\ket{k, 1} 
&= L_{1,n+1} (\ket{2^n-1-k} \otimes \ket{0}) \\
&= L_{1,n+1} \ket{2(2^n-1-k)} \\
&= \ket{2(2^n-1-k) + 1} \\
&= \ket{2^n-1-k} \otimes \ket{1}\\
&=J_{n}\ket{k}\otimes\ket{1}    
\end{split}
\end{equation}

Combining these cases, the relation $U_{n+1} = U_n \otimes \ket{0}\!\bra{0} + J_n \otimes \ket{1}\!\bra{1}$ holds for all $n \ge 1$.
\end{proof}

\begin{remark}[From structural recursion to carry recursion]
The recursion of $U_n$ should be interpreted as a structural recursion for the reflected generator $\widetilde{L}_{1,n}=L_{1,n}J_n=U_n$, not as the preferred arithmetic implementation of $L_{1,n}$ itself. The usual binary carry rule is recovered only after composing back with the reversal layer through $L_{1,n+1}=U_{n+1}J_{n+1}$. Under little-endian ordering, this converts the block recursion
\begin{equation}
U_{n+1}=U_n\otimes\ket{0}\!\bra{0}+J_n\otimes\ket{1}\!\bra{1}
\end{equation}
into the standard increment logic for $L_{1,n+1}$: an original input LSB of $0$ produces no carry to higher bits, whereas an original input LSB of $1$ triggers recursive carry propagation. Thus the structural recursion of $U_n$ is the operator-level source, while the compiled increment recursion described in Sec.~\ref{sec:complexity} is its arithmetic carry-propagation form after multiplication by $J_{n+1}$.
\end{remark}

Figure~\ref{fig:qc_recursive} illustrates the direct structural-recursive realization associated with the recursion of the reflected generator $U_n$.
\begin{figure}[htbp]
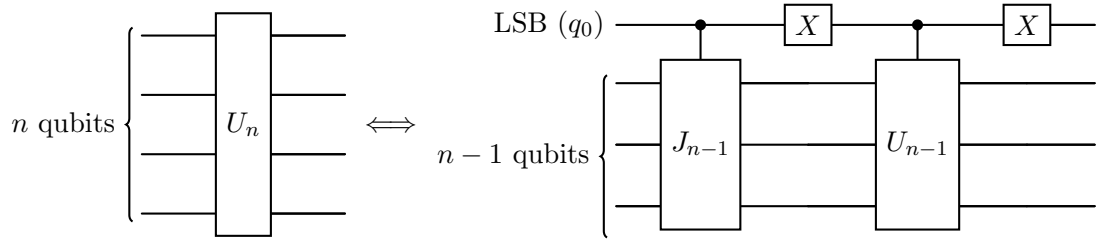

\centering
\resizebox{0.9\columnwidth}{!}{%
\begin{tabular}{ccc}
\input{figs/qc_Un}&
$\Longleftrightarrow$ &
\input{figs/qc_recursive}\\
\end{tabular}
}
\caption{Recursive construction of the reflected generator $U_n$. The operation on the upper register is conditioned on the LSB ($q_0$): if $q_0=\ket{0}$, $U_{n-1}$ is applied; if $q_0=\ket{1}$, $J_{n-1}$ is applied.}
\label{fig:qc_recursive}
\end{figure}

\section{Pauli Support of the Reflected-Shift Family}
\label{app:pauli_support}
The structural recursion of $U_n$ also induces a corresponding recursion for the full reflected-shift family $\widetilde{L}_{i,n}$, which in turn sharply constrains its Pauli support.
\begin{proposition}[Pauli support of the reflected-shift family]
\label{prop:pauli_support_reflected}
For each $n\ge 1$, every reflected shift $\widetilde{L}_{i,n}=L_{i,n}J_n$ lies in
\begin{equation}
\mathcal S_n:=\mathrm{span}\!\left(\{I,X\}\otimes\{I,X,Z\}^{\otimes(n-1)}\right).
\end{equation}
In particular, no Pauli-$Y$ term appears in the Pauli decomposition of any reflected shift. Hence the family uses at most
\begin{equation}
|\mathcal S_n| = 2\cdot 3^{n-1}=\frac{2}{3} \cdot 3^n
\end{equation}
distinct Pauli strings.
\end{proposition}

\begin{proof}
We argue by induction on $n$. For $n=1$, one has
\begin{equation}
\widetilde{L}_{0,1}=X,\qquad \widetilde{L}_{1,1}=I,
\end{equation}
so the claim holds with $\mathcal S_1=\mathrm{span}\{I,X\}$. For the inductive step, assume that for some fixed $n\ge 1$, every reflected shift $\widetilde{L}_{i,n}$ with $i\in\{0,\dots,2^n-1\}$ lies in 
\begin{equation}
\mathcal S_n=
\mathrm{span}\!\left(\{I,X\}\otimes\{I,X,Z\}^{\otimes n-1}\right).
\end{equation}
Now consider $n+1$ qubits and write the index as either $2r$ or $2r+1$, with $r\in\{0,\dots,2^n-1\}$. We use the main text definition $\widetilde{L}_{i,n}=L_{i,n}J_n$, the recursion of $U_n$ from Appendix~\ref{sec:recursive_un}, and the shift decomposition from Eq.~\eqref{eq:pow2_shift}. In particular, on $(n+1)$ little-endian qubits, an even shift leaves the LSB unchanged and acts as a shift by $r$ on the upper $n$ qubits, so
\begin{equation}
L_{2r,n+1}=L_{r,n}\otimes I.
\end{equation}
Hence
\begin{equation}
\begin{split}
\widetilde{L}_{2r,n+1}
&=L_{2r,n+1}J_{n+1} \\
&=(L_{r,n}\otimes I)(J_n\otimes X) \\
&=(L_{r,n}J_n)\otimes X \\
&=\widetilde{L}_{r,n}\otimes X.
\end{split}
\end{equation}

For the odd branch, using $L_{2r+1,n+1}=L_{2r,n+1}L_{1,n+1}$ together with $L_{1,n+1}J_{n+1}=U_{n+1}$, we obtain
\begin{equation}
\begin{split}
\widetilde{L}_{2r+1,n+1}
&=L_{2r+1,n+1}J_{n+1} \\
&=L_{2r,n+1}U_{n+1} \\
&=(L_{r,n}\otimes I)\left(U_n\otimes\ket{0}\!\bra{0}+J_n\otimes\ket{1}\!\bra{1}\right) \\
&=L_{r,n}U_n\otimes\ket{0}\!\bra{0}+L_{r,n}J_n\otimes\ket{1}\!\bra{1} \\
&=\widetilde{L}_{r+1,n}\otimes\ket{0}\!\bra{0}
+\widetilde{L}_{r,n}\otimes\ket{1}\!\bra{1},
\end{split}
\end{equation}
where the index $r+1$ is taken modulo $2^n$, so that $r=2^n-1$ gives $r+1\equiv 0 \pmod{2^n}$.

By the induction hypothesis, both operators appearing on the right-hand side, namely $\widetilde{L}_{r,n}$ and $\widetilde{L}_{r+1,n}$, lie in $\mathcal S_n$. From Eq.~\eqref{eq:projectors}
it follows that both even and odd branches lie in
\begin{equation}
\mathcal S_n\otimes\mathrm{span}\{I,X,Z\}
=
\mathrm{span}\!\left(\{I,X\}\otimes\{I,X,Z\}^{\otimes n}\right)
=
\mathcal S_{n+1}.
\end{equation}
Thus every reflected shift $\widetilde{L}_{i,n+1}$ lies in $\mathcal S_{n+1}$, and no Pauli-$Y$ term can occur.

Finally, the number of basis strings in $\mathcal S_n$ is
\begin{equation}
|\mathcal S_n| = 2\cdot 3^{n-1},
\end{equation}
since the leftmost tensor factor contributes two possibilities $\{I,X\}$ and each remaining tensor factor contributes three possibilities $\{I,X,Z\}$.
\end{proof}

Because $H_n(\bm b)=\sum_{i=0}^{2^n-1} b_i\,\widetilde{L}_{i,n}$, the convolution operator itself lies in the same Pauli-support subspace $\mathcal S_n$. Consequently, its Pauli decomposition contains no $Y$ terms and uses at most $2\cdot 3^{n-1}$ strings, with generic coefficients $\bm b$ typically populating this support up to accidental cancellations.

\section{Bitwise Compilation of the Reflected-Shift SELECT Block}
\label{app:bitwise_compilation}
This appendix records the supplementary optimized bitwise recursive compilation underlying the compiled macro block estimates used in Sec.~\ref{sec:complexity}, proves its exact equivalence to the recursive construction used in the main text, and records the two explicit recursive realizations of the block $W_{A_m}[L_{2^m,n}]$ used in Algorithm~\ref{alg:quantum_conv}.

\myheading{Conventions.}
\label{app:compiled_conventions}
We keep the notation of the main text and only introduce the compiled suffix-increment block used in Sec.~\ref{sec:complexity}.
Let $N=2^n$. We use little-endian integer encoding on the ancilla register $A=(A_0, \cdots, A_{n-1})$ and the data register
$D=(D_0,\dots,D_{n-1})$:
\begin{equation}
x=\sum_{j=0}^{n-1} x_j 2^j,\qquad x_j\in\{0,1\}.
\end{equation}
Here $D_0$ is the LSB.
As in the main text,
\begin{equation}
L_{r,n}\ket{x}=\ket{x+r \pmod{2^n}},\qquad
J_n=X^{\otimes n}.
\end{equation}
We also use the main text definition of $\mathrm{SELECT}_{\widetilde{L}}$ in Eq.~\eqref{eq:select_tildeL_coherent}.
To state the compiled network unambiguously, we use the multi-control extension of Eq.~\eqref{eq:controlled_W_def} with $r$ control qubits $q_1, \cdots q_r$:
\begin{equation}
W_{q_1,\dots,q_r}[O]
:=
\left(I-\Pi_{q_1,\dots,q_r}\right)\otimes I
+
\Pi_{q_1,\dots,q_r}\otimes O,\qquad
\Pi_{q_1,\dots,q_r}:=\bigotimes_{\ell=1}^{r}\ket{1}\!\bra{1}_{q_\ell},
\label{eq:multi_control_W}
\end{equation}
Our goal is to explicitly construct a gate-level circuit that increments the data subregister from bit $m$ to $n-1$.
As in the compiled macro block discussion of Sec.~\ref{sec:complexity}, this is realized by a standard carry-propagation cascade.
Accordingly, for each $m\in\{0,\dots,n-1\}$, we define this compiled incrementer block as
\begin{equation}
\mathrm{INC}_{[m:n-1]}^{\mathrm{(cmp)}}
:=X_{D_m}\left(\prod_{j=m+1}^{n-1}W_{D_m,\dots,D_{j-1}}[X_{D_j}]\right),
\label{eq:inc_compiled_def}
\end{equation}
where, under the standard right-to-left action on kets, the gates are applied in descending target order
$j=n-1,n-2,\dots,m+1$, followed by $X_{D_m}$.
Its $A_m$-controlled form is (see Fig.~\ref{fig:qc_bitwise_compiled})
\begin{equation}
W_{A_m}\!\left[\mathrm{INC}_{[m:n-1]}^{\mathrm{(cmp)}}\right]
=
W_{A_m}[X_{D_m}]
\left(\prod_{j=m+1}^{n-1}W_{A_m,D_m,\dots,D_{j-1}}[X_{D_j}]\right).
\label{eq:inc_compiled_controlled}
\end{equation}

\myheading{Explicit recursion of the compiled suffix increment.}
\label{app:compiled_suffix_recursion}
The compiled suffix-increment network also admits a direct recursive description.
For $m=n-1$ we have the base case
\begin{equation}
\mathrm{INC}_{[n-1:n-1]}^{\mathrm{(cmp)}} = X_{D_{n-1}}.
\end{equation}
For $m=0,\dots,n-2$, the compiled suffix increment satisfies
\begin{equation}
\mathrm{INC}_{[m:n-1]}^{\mathrm{(cmp)}}
=
X_{D_m}\, W_{D_m}\!\left[\mathrm{INC}_{[m+1:n-1]}^{\mathrm{(cmp)}}\right].
\label{eq:inc_compiled_recursion}
\end{equation}
Equivalently, its $A_m$-controlled form obeys
\begin{equation}
W_{A_m}\!\left[\mathrm{INC}_{[m:n-1]}^{\mathrm{(cmp)}}\right]
=
W_{A_m}[X_{D_m}]\,
W_{A_m,D_m}\!\left[\mathrm{INC}_{[m+1:n-1]}^{\mathrm{(cmp)}}\right],
\label{eq:inc_compiled_controlled_recursion}
\end{equation}
with base case
\begin{equation}
W_{A_m}\!\left[\mathrm{INC}_{[n-1:n-1]}^{\mathrm{(cmp)}}\right]
=
W_{A_m}[X_{D_{n-1}}].
\end{equation}
Expanding Eq.~\eqref{eq:inc_compiled_controlled_recursion} recursively reproduces exactly the multi-controlled-$X$ chain in Eq.~\eqref{eq:inc_compiled_controlled}. Thus the compiled bitwise realization can itself be generated recursively, but in a gate-optimized carry-propagation form.

\begin{figure}[t]
\centering
\input{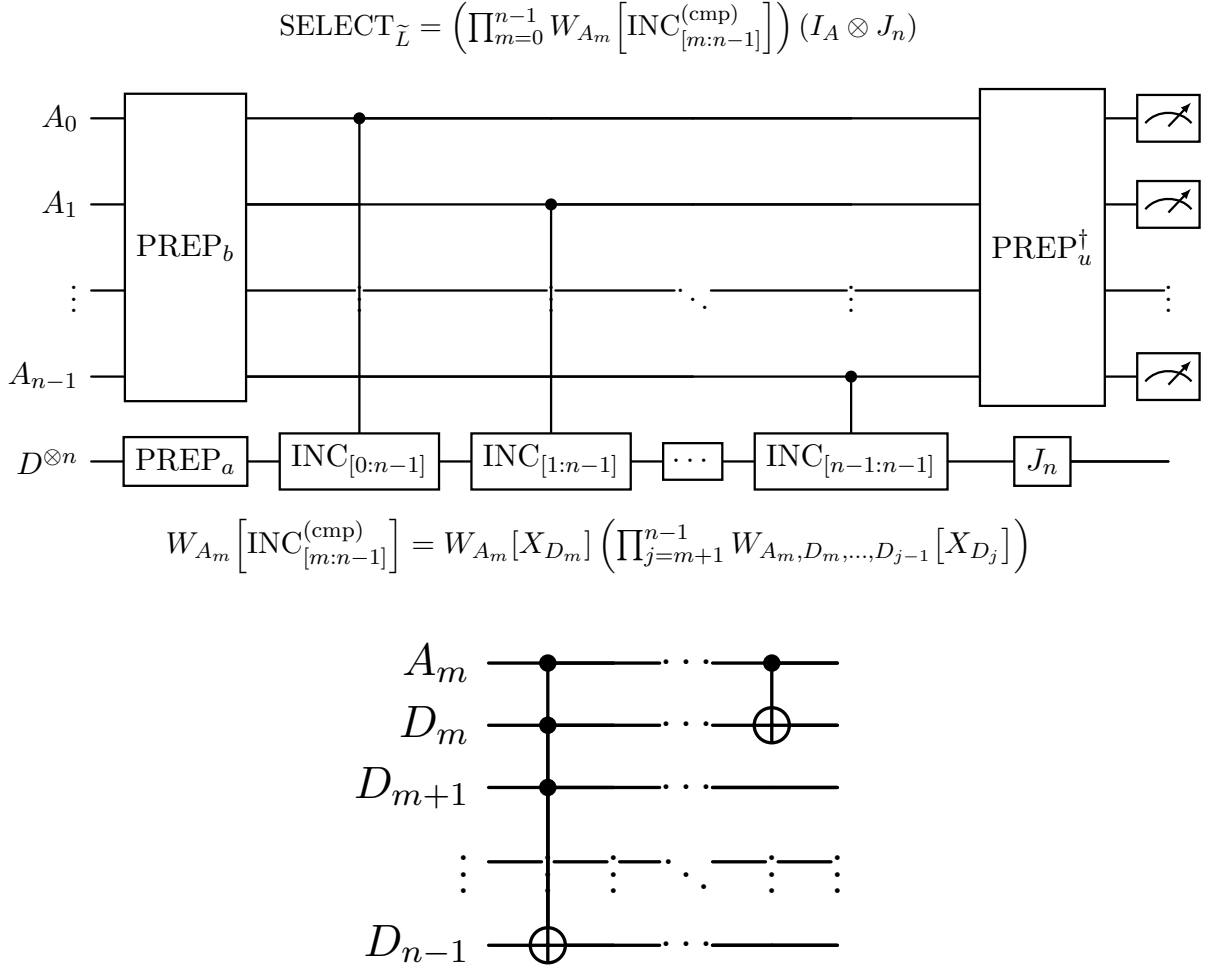}
\caption{Compiled $n$-qubit bitwise realization with explicit internal structure. Top panel: $\mathrm{SELECT}_{\widetilde{L}}=\left(\prod_{m=0}^{n-1}W_{A_m}\!\left[\mathrm{INC}_{[m:n-1]}^{\mathrm{(cmp)}}\right]\right)(I_A\otimes J_n)$ inside the asymmetric-LCU pipeline. Bottom panel: decomposition of one block $W_{A_m}\!\left[\mathrm{INC}_{[m:n-1]}^{\mathrm{(cmp)}}\right]$ into a chain of multi-controlled $X$ operations $W_{A_m,D_m,\ldots,D_{j-1}}[X_{D_j}]$ ending with $W_{A_m}[X_{D_m}]$.}
\label{fig:qc_bitwise_compiled}
\end{figure}

\begin{lemma}[Compiled bitwise controlled increment]
\label{lem:compiled_inc}
Fix $k\ge1$, and let
\begin{equation}
x=\sum_{j=0}^{k-1} x_j 2^j,\qquad x_j\in\{0,1\},
\end{equation}
so that $\ket{x}=\ket{x_{k-1}\cdots x_0}$ with $x_0$ the LSB, consistent with the conventions introduced in Appendix~\ref{app:compiled_conventions}.
Let $c\in\{0,1\}$.
Consider the $k$-bit version of Eq.~\eqref{eq:inc_compiled_controlled}: for each $j=1,\dots,k-1$, apply the
multi-controlled $X$ gate with controls $c,x_0,\dots,x_{j-1}$ and target bit $j$.
Under the right-to-left convention these gates are applied in descending order $j=k-1,k-2,\dots,1$,
followed by the final controlled $X$ on the LSB.
Then
\begin{equation}
\ket{c}\ket{x}\longmapsto \ket{c}\ket{x+c \pmod{2^k}}.
\end{equation}
Equivalently, the output bits satisfy
\begin{equation}
y_0=x_0\oplus c,\qquad
y_j=x_j\oplus\Bigl(c\land \bigwedge_{\ell=0}^{j-1}x_\ell\Bigr)\quad (j=1,\dots,k-1).
\end{equation}
Hence the circuit is exactly $W_c[L_{1,k}]$.
\end{lemma}

\begin{proof}
If $c=0$, no controlled operation is applied, and the register remains unchanged.
Now assume $c=1$.
Because the targets are visited in descending order $j=k-1,\dots,1$, when the gate targeting bit $j$ is applied,
all lower bits $x_0,\dots,x_{j-1}$ are still their original input values.
Therefore that gate flips bit $j$ if and only if all lower bits are $1$.
After the final flip of the LSB, the output bits are
\begin{equation}
y_0=x_0\oplus 1,\qquad
y_j=x_j\oplus\Bigl(\bigwedge_{\ell=0}^{j-1}x_\ell\Bigr)\quad (j=1,\dots,k-1).
\end{equation}
These are exactly the usual binary carry rules for $x\mapsto x+1 \pmod{2^k}$:
bit $j$ flips precisely when a carry reaches it, and that happens exactly when all less-significant bits are $1$.
Combining the cases $c=0$ and $c=1$ gives
\begin{equation}
y_0=x_0\oplus c,\qquad
y_j=x_j\oplus\Bigl(c\land \bigwedge_{\ell=0}^{j-1}x_\ell\Bigr)\quad (j=1,\dots,k-1),
\end{equation}
hence $\ket{c}\ket{x}\mapsto \ket{c}\ket{x+c \pmod{2^k}}$.
\end{proof}

\begin{lemma}[Suffix increment realizes the block shift $L_{2^m,n}$]
\label{lem:suffix_shift}
Let $m\in\{0,\dots,n-1\}$ and $k=n-m$.
Applying $L_{1,k}$ to suffix qubits $(D_m,\dots,D_{n-1})$ is
\begin{equation}
L_{1,n-m}\otimes I^{\otimes m}=L_{2^m,n}.
\label{eq:suffix_weight}
\end{equation}
\end{lemma}

\begin{proof}
Write
\begin{equation}
x=x_{\mathrm{low}}+2^m x_{\mathrm{high}},
\quad
0\le x_{\mathrm{low}}<2^m,\quad
0\le x_{\mathrm{high}}<2^{n-m}.
\end{equation}
Suffix increment sends $x_{\mathrm{high}}\mapsto x_{\mathrm{high}}+1 \pmod{2^{n-m}}$.
Hence
\begin{equation}
x' = x_{\mathrm{low}}+2^m(x_{\mathrm{high}}+1)
\equiv x+2^m \pmod{2^n},
\end{equation}
which is exactly $L_{2^m,n}$.
\end{proof}

\begin{proposition}[Exact equivalence to the compiled realization of $\mathrm{SELECT}_{\widetilde{L}}$]
\label{prop:exact_equiv}
The compiled bitwise circuit
\begin{equation}
\left(\prod_{m=0}^{n-1}W_{A_m}\!\left[\mathrm{INC}_{[m:n-1]}^{\mathrm{(cmp)}}\right]\right)(I_A\otimes J_n)
=
\mathrm{SELECT}_{\widetilde{L}}.
\end{equation}
\end{proposition}

\begin{proof}
By Eq.~\eqref{eq:inc_compiled_controlled} and Lemma~\ref{lem:compiled_inc},
\begin{equation}
W_{A_m}\!\left[\mathrm{INC}_{[m:n-1]}^{\mathrm{(cmp)}}\right]
=
W_{A_m}[L_{1,n-m}\otimes I^{\otimes m}].
\end{equation}
By Lemma~\ref{lem:suffix_shift},
\begin{equation}
W_{A_m}\!\left[\mathrm{INC}_{[m:n-1]}^{\mathrm{(cmp)}}\right]
=
W_{A_m}[L_{2^m,n}]
\end{equation}
for every $m$. Therefore
\begin{equation}
\left(\prod_{m=0}^{n-1}W_{A_m}\!\left[\mathrm{INC}_{[m:n-1]}^{\mathrm{(cmp)}}\right]\right)(I_A\otimes J_n)
=
\left(\prod_{m=0}^{n-1}W_{A_m}[L_{2^m,n}]\right)(I_A\otimes J_n)
=
\mathrm{SELECT}_{\widetilde{L}},
\end{equation}
where the last equality is Eq.~\eqref{eq:select_tildeL_coherent}. This is block-by-block equality of the construction, not only input-output equivalence.
\end{proof}

\myheading{Relation to the recursive block.}
\label{app:compiled_recursive_relation}
The compiled increment recursion is not an independent starting point; it is the arithmetic carry-propagation form of the same unit shift $L_{1,n}$ that appears in the main text through the reflected generator relation
\begin{equation}
L_{1,n}=U_nJ_n.
\end{equation}
Thus, for each $m$,
\begin{equation}
\mathrm{INC}_{[m:n-1]}^{\mathrm{(cmp)}}
=
L_{1,n-m}\otimes I^{\otimes m}
=
\left(U_{n-m}J_{n-m}\right)\otimes I^{\otimes m}.
\end{equation}
In particular, for $m=0$,
\begin{equation}
\mathrm{INC}_{[0:n-1]}^{\mathrm{(cmp)}}=U_nJ_n=L_{1,n}.
\end{equation}
Accordingly, the main text recursion of $U_n$ should be viewed as a structural normal form for the reflected generator, whereas the compiled recursion of $\mathrm{INC}$ is the corresponding binary carry-propagation realization of the same shift action. See Fig.~\ref{fig:qc_bitwise_compiled}.

\myheading{Two recursive realizations of the controlled shift block.}
\label{app:shiftblock_recursions}

\begin{algorithm}[H]
    \caption{Direct recursive realization of $W_{A_m}[L_{2^m,n}]$}
    \label{alg:direct_recursive_block}
    \begin{algorithmic}[1]
        \Procedure{\textsc{Synthesize-U}}{$k;\,c_1,\dots,c_r;\,q_0,\dots,q_{k-1}$}
            \If{$k=1$}
                \State \Return
            \EndIf
            \State Apply $X_{q_0}$
            \State \Call{\textsc{Synthesize-U}}{$k-1;\,c_1,\dots,c_r,q_0;\,q_1,\dots,q_{k-1}$}
            \State Apply $X_{q_0}$
            \For{$j=1$ \textbf{to} $k-1$}
                \State Apply $W_{c_1,\dots,c_r,q_0}[X_{q_j}]$
            \EndFor
        \EndProcedure

        \Procedure{\textsc{Apply-Direct-Block}}{$A_m;\,D_m,\dots,D_{n-1}$}
            \For{$j=m$ \textbf{to} $n-1$}
                \State Apply $W_{A_m}[X_{D_j}]$
            \EndFor
            \State \Call{\textsc{Synthesize-U}}{$n-m;\,A_m;\,D_m,\dots,D_{n-1}$}
        \EndProcedure
    \end{algorithmic}
\end{algorithm}

In Algorithm~\ref{alg:quantum_conv}, the only realization-dependent ingredient is the controlled block shift $W_{A_m}[L_{2^m,n}]$. We record below two recursive realizations of $W_{A_m}[L_{2^m,n}]$. Algorithm~\ref{alg:direct_recursive_block} gives the direct structural-recursive realization induced by the reflected-generator recursion of Appendix~\ref{app:Un_recursion}. Algorithm~\ref{alg:bitwise_recursive_block} gives the optimized bitwise recursive compilation in carry-propagation form. By Proposition~\ref{prop:exact_equiv}, both realize this target block.

\begin{algorithm}[H]
    \caption{Bitwise recursive realization of $W_{A_m}[L_{2^m,n}]$}
    \label{alg:bitwise_recursive_block}
    \begin{algorithmic}[1]
        \Procedure{\textsc{Apply-Compiled-Inc}}{$k;\,c;\,q_0,\dots,q_{k-1}$}
            \If{$k=1$}
                \State Apply $W_c[X_{q_0}]$
                \State \Return
            \EndIf
            \State Apply $W_{c,q_0,\dots,q_{k-2}}[X_{q_{k-1}}]$
            \State \Call{\textsc{Apply-Compiled-Inc}}{$k-1;\,c;\,q_0,\dots,q_{k-2}$}
        \EndProcedure

        \Procedure{\textsc{Apply-Compiled-Block}}{$A_m;\,D_m,\dots,D_{n-1}$}
            \State \Call{\textsc{Apply-Compiled-Inc}}{$n-m;\,A_m;\,D_m,\dots,D_{n-1}$}
        \EndProcedure
    \end{algorithmic}
\end{algorithm}

\bibliography{refs}

\end{document}